\documentclass[a4paper,USenglish,cleveref,autoref,thm-restate]{lipics-v2019}
\nolinenumbers

\usepackage{microtype}
\usepackage{xspace}
\usepackage{tikz}
\usepackage{cleveref}
\usepackage{hyperref}
\renewcommand{\cref}{\Cref}

\usepackage[normalem]{ulem}

\usepackage{thm-restate}

\Crefname{lemma}{Lemma}{Lemmas}
\crefname{lemma}{Lemma}{Lemmas}

\newcommand{\cola}{\ensuremath{\mathsf{a}}\xspace}
\newcommand{\colb}{\ensuremath{\mathsf{b}}\xspace}
\newcommand{\colc}{\ensuremath{\mathsf{c}}\xspace}
\newcommand{\cold}{\ensuremath{\mathsf{d}}\xspace}
\newcommand{\cole}{\ensuremath{\mathsf{e}}\xspace}

\newcommand{\Q}{\ensuremath{\mathcal{Q}}\xspace}
\newcommand{\N}{\ensuremath{\mathbb{N}}\xspace}
\newcommand{\OR}{\textsc{or}\xspace}

\newcommand{\notcontainment}{\ensuremath{\mathsf{NP \not\subseteq coNP/poly}}\xspace}

\newcommand{\eqvr}[0]{\ensuremath{\mathcal{R}}\xspace}
\newcommand{\containment}{\ensuremath{\mathsf{NP  \subseteq coNP/poly}}\xspace}

\newcommand{\TODO}[1][]{%
  \ifx/#1/%
    \textcolor{red}{TODO!}%
  \else%
    \textcolor{red}{todo: #1}%
  \fi%
}
\newcommand{\Oh}{\ensuremath{\mathcal{O}}\xspace}

\newcommand{\Hcoloring}{\textsc{$H$-Coloring}\xspace}
\newcommand{\listHcoloring}{\textsc{List $H$-Coloring}\xspace}
\newcommand{\listcoloring}[1]{\textsc{List #1-Coloring}\xspace}
\newcommand{\Anlistcoloring}{\textsc{Annotated List $P_4$-Coloring}\xspace}
\newcommand{\ELL}{L}

\newcommand{\cA}{\mathcal{A}}
\newcommand{\cB}{\mathcal{B}}
\newcommand{\cC}{\mathcal{C}}
\newcommand{\cD}{\mathcal{D}}
\newcommand{\cE}{\mathcal{E}}

\newcommand{\cK}{\mathcal{K}}

\newcommand{\cP}{\mathcal{P}}
\newcommand{\cQ}{\mathcal{Q}}
\newcommand{\cR}{\mathcal{R}}

\newcommand{\ovp}{\overline{P}}
\newcommand{\pstar}{P^*}

\usepackage{todonotes}

\newcommand{\defproblem}[3]{\par
 \vspace{3mm}
\noindent\fbox{
 \begin{minipage}{0.96\textwidth}
 \begin{tabular*}{\textwidth}{@{\extracolsep{\fill}}lr} #1 &  \vspace{1mm} \\ \end{tabular*}
{\textbf{Input:}} #2
  \vspace{1mm}\\%
 {\textbf{Question:}} #3
 \end{minipage}
 }
 \vspace{3mm}\par
}

\renewcommand{\tilde}{\widetilde}

\title{Sparsification Lower Bounds for List $H$-Coloring}

\author{Hubie Chen}{Birkbeck, University of London,
Malet Street, Bloomsbury,
London WC1E 7HX, United Kingdom}{hubie@dcs.bbk.ac.uk}{}{}
\author{Bart M.\,P. Jansen}{Eindhoven University of Technology, P.O. Box 513, 5600 MB Eindhoven, The Netherlands}{b.m.p.jansen@tue.nl}{https://orcid.org/0000-0001-8204-1268}{Supported by NWO Gravitation grant ``Networks''.}
\author{Karolina Okrasa}{University of Warsaw, Institute of Informatics \\ \& Warsaw University of Technology, Faculty of Mathematics and Information Science}{k.okrasa@mini.pw.edu.pl}{https://orcid.org/0000-0003-1414-3507}{Supported by the European Research Council (ERC) under the European
Union’s Horizon 2020 research and innovation programme Grant Agreement no. 714704.}
\author{Astrid Pieterse}{Department of Computer Science, Humboldt-Universit{\"a}t zu Berlin, Germany}{astrid.pieterse@informatik.hu-berlin.de}{https://orcid.org/0000-0003-3721-6721}{Supported by DFG Emmy Noether-grant (KR 4286/1), part of this research was supported by NWO Gravitation grant ``Networks''.}
\author{Pawe\l{} Rz\k{a}\.zewski}{Warsaw University of Technology, Faculty of Mathematics and Information Science\\
\& University of Warsaw, Institute of Informatics}{p.rzazewski@mini.pw.edu.pl}{https://orcid.org/0000-0001-7696-3848}{Supported by Polish National Science Centre grant no. 2018/31/D/ST6/00062.}
\authorrunning{H. Chen, B.\,M.\,P. Jansen, K. Okrasa, A. Pieterse, and P. Rz\k{a}\.zewski}
\Copyright{Hubie Chen, Bart M.\,P. Jansen, Karolina Okrasa, Astrid Pieterse, and Pawe\l{} Rz\k{a}\.zewski}
\category{}
\relatedversion{The extended abstract of this paper was accepted to ISAAC 2020.}
\supplement{}
\funding{}

\acknowledgements{We acknowledge the welcoming and productive atmosphere at Dagstuhl Seminar~19271 ``Graph Colouring: from Structure to Algorithms'', where this work has been initiated.}
\EventEditors{John Q. Open and Joan R. Access}
\EventNoEds{2}
\EventLongTitle{}
\EventShortTitle{}
\EventAcronym{}
\EventYear{}
\EventDate{}
\EventLocation{}
\EventLogo{}
\SeriesVolume{}
\ArticleNo{1}
\hideLIPIcs
\date{}
\begin{document}

\maketitle

\begin{abstract}
\normalsize
We investigate the \textsc{List $H$-Coloring} problem, the generalization of graph coloring that asks whether an input graph~$G$ admits a homomorphism to the undirected graph~$H$ (possibly with loops), such that each vertex~$v \in V(G)$ is mapped to a vertex on its list~$L(v) \subseteq V(H)$. An important result by Feder, Hell, and Huang [JGT~2003] states that \textsc{List $H$-Coloring} is polynomial-time solvable if~$H$ is a so-called \emph{bi-arc graph}, and NP-complete otherwise. We investigate the NP-complete cases of the problem from the perspective of polynomial-time sparsification: can an $n$-vertex instance be efficiently reduced to an equivalent instance of bitsize~$\Oh(n^{2-\varepsilon})$ for some~$\varepsilon > 0$? We prove that if~$H$ is not a bi-arc graph, then \textsc{List $H$-Coloring} does not admit such a sparsification algorithm unless \containment. Our proofs combine techniques from kernelization lower bounds with a study of the structure of graphs~$H$ which are not bi-arc graphs.
\end{abstract}

\clearpage

\section{Introduction} \label{introduction}

\subparagraph*{Background and motivation}
The \textsc{List $H$-Coloring} problem is a generalization of the classic graph coloring problem. For a fixed undirected graph~$H$, possibly with self-loops, an input to the problem consists of an undirected graph~$G$ together with a list~$L(v) \subseteq V(H)$ for each vertex~$v \in V(G)$. The question is whether there is a list homomorphism from~$G$ to~$H$: a mapping~$f \colon V(G) \to V(H)$ such that~$\{f(u), f(v)\} \in E(H)$ for all~$\{u,v\} \in E(G)$, and such that~$f(v) \in L(v)$ for all~$v \in V(G)$. When~$H$ is a $q$-clique and~$L(v) = V(H)$ for each vertex, \textsc{List $H$-Coloring} is equivalent to traditional graph $q$-colorability.

The classic computational complexity of \textsc{List $H$-Coloring} for other graphs~$H$ has been investigated, next to a long line of work for the non-list version of the problem~\cite{AlbertsonCG85,Bulatov05,HellN90,Irving83,KunS16,MaurerSW81,Nesetril81,Siggers10}.
As the first step towards the dichotomy, Feder and Hell~\cite{FederH98} proved that if $H$ is reflexive (i.e., every vertex has a self-loop), then \listHcoloring is polynomial-time solvable if $H$ is an interval graph, and NP-complete otherwise.
Next, a dichotomy for irreflexive graphs $H$ was proven by Feder, Hell, and Huang~\cite{Feder99}:  the problem is polynomial-time solvable if $H$ is bipartite and additionally its complement is a circular-arc graph, and in all other cases the problem is NP-complete. It is interesting to mention that the subclass of bipartite graphs consisting of those which are complements of circular-arc graphs, was already studied by Trotter and Moore in the context of classifying some posets~\cite{TROTTER1976361}.
Finally, Feder, Hell, and Huang~\cite{FederHH03} defined a new class of geometric intersection graphs (potentially with loops), called \emph{bi-arc graphs}, which encapsulates reflexive interval graphs and (irreflexive) bipartite co-circular-arc graphs. We postpone the definition of bi-arc graphs to \cref{sec:biarc}. Feder, Hell, and Huang proved a powerful dichotomy theorem: \textsc{List $H$-Coloring} is polynomial-time solvable if~$H$ is a bi-arc graph, but NP-complete otherwise.

In this work we investigate \textsc{List $H$-Coloring} from the perspective of polynomial-time sparsification (cf.~\cite{ChenJP18,DellM14,JansenP17a}). From this viewpoint, the goal is to develop a polynomial-time algorithm that maps a (potentially dense) $n$-vertex instance~$G$ to a smaller instance~$G'$ that can be encoded in~$f(n)$ bits for some size function~$f$, yet which has the same \textsc{yes}/\textsc{no} answer as~$G$. Observe that this is trivial if~$f(n) = n^2$; we refer to a sparsification algorithm as \emph{nontrivial} if it achieves a size bound of~$f(n) \in \Oh(n^{2-\varepsilon})$ bits for some~$\varepsilon > 0$.

The general quest for sparsification algorithms is motivated by the fact that they allow instances to be stored, manipulated, and solved more efficiently: since sparsification preserves the exact answer to the problem, it suffices to solve the sparsified instance. Our interest in sparsification for \textsc{List $H$-Coloring} has a number of motivations, which we now describe.

There is a growing list of problems for which the existence of nontrivial sparsification algorithms has been ruled out under the established assumption \notcontainment, which includes \textsc{Vertex Cover}~\cite{DellM14}, \textsc{Dominating Set}~\cite{JansenP17a}, \textsc{Feedback Arc Set}~\cite{JansenP17a}, and \textsc{Treewidth}~\cite{Jansen15}. To the best of our knowledge, to date there is no non-trivial sparsification algorithm for \emph{any} NP-hard problem that is defined on general graphs. Could it be that there is \emph{no} natural NP-hard graph problem that admits a nontrivial sparsification algorithm? The surprising richness of problems that admit a polynomial kernelization, a desirable outcome in a different regime of efficient preprocessing (cf.~\cite{FominLSZ18,GuoN07}), may tempt one to believe that for the right problem, something nontrivial can be done. In an attempt to identify a problem that admits nontrivial sparsification, we target the broad class of \textsc{List $H$-Coloring} decision problems.

A second motivation for studying \textsc{List $H$-Coloring} comes from its interpretation as a constraint satisfaction problem: an instance of \textsc{List $H$-Coloring} corresponds to a CSP that has a variable for each vertex of the input graph~$G$, which has to be assigned a value from the set~$V(H)$. For each edge~$\{u,v\}$ of~$G$ there is a constraint that the value assigned to~$u$ should be a neighbor (in graph~$H$) of the value assigned to~$v$, and for each vertex~$v \in V(G)$ there is a constraint that the value of~$v$ belongs to~$L(v)$. Hence any NP-hard \textsc{List $H$-Coloring} problem translates into a CSP with a non-Boolean domain in which constraints have arity at most two. Recent work~\cite{ChenJP18,LagerkvistW17} has led to a number of nontrivial advances in the study of sparsification for CSPs with a Boolean domain. A natural next step in that line of research is to target non-Boolean CSPs, of which the \textsc{List $H$-Coloring} problems form a rich subset.

The last motivation for studying sparsification for \textsc{List $H$-Coloring} is that it forms the logical next step in the study of sparsification for coloring problems. Recent work~\cite{JansenP17} showed that \textsc{Graph (List) $q$-Colorability} does not admit nontrivial polynomial-time sparsification for~$q\geq 3$ unless \containment, but left the case of \textsc{List $H$-Coloring} open.

\subparagraph*{Our results}
We prove that for all undirected, possibly non-simple, graphs~$H$ for which \textsc{List $H$-Coloring} is NP-complete, the problem does not admit nontrivial sparsification unless an unlikely complexity-theoretic collapse occurs. Our proofs combine techniques from kernelization lower bounds with a careful analysis of the common structures of hard graphs~$H$. To state our sparsification lower bounds in full generality, we use the notion of generalized kernelization (see Definition~\ref{def:generalized:kernel}), where the number of vertices~$n$ of the instance plays the role of the complexity parameter~$k$. A generalized kernelization for \textsc{List $H$-Coloring} of size~$f(n)$ is therefore a polynomial-time algorithm that maps any $n$-vertex input~$G$, to an equivalent instance (of a potentially different but fixed decision problem) of bitsize~$f(n)$. Since a polynomial-time sparsification algorithm mapping to instances of bitsize~$f(n)$ yields a generalized kernelization of size~$f(n)$, lower bounds on the latter also apply to the former.

%

\begin{restatable}{theorem}{theoremlowerbound}\label{theorem:lowerbound}
If~$H$ is an undirected graph that is not a bi-arc graph, possibly with loops, then \textsc{List $H$-Coloring} parameterized by the number of vertices~$n$ admits no generalized kernel of size~$\Oh(n^{2-\varepsilon})$ for any~$\varepsilon > 0$, unless \containment.
\end{restatable}

The techniques employed in the proof of Theorem~\ref{theorem:lowerbound} are rather different from those in the NP-completeness proof for the hard cases of \textsc{List $H$-Coloring}. Feder, Hell, and Huang~\cite{FederHH03} establish the NP-completeness of \textsc{List $H$-Coloring} when~$H$ is not a bi-arc graph, by reducing from \textsc{$3$-Coloring}. They build gadgets in \textsc{List $H$-Coloring} instances to mimic the effect of a normal edge in \textsc{$3$-Coloring}, and then replace each edge with such a gadget. Although \textsc{$3$-Coloring} is known not to admit any nontrivial sparsification unless \containment~\cite{JansenP19}, the mentioned NP-completeness reduction does not transfer this lower bound from \textsc{$3$-Coloring} to \textsc{List $H$-Coloring}: as the reduction introduces a gadget (with new vertices) for every edge of the \textsc{$3$-Coloring} instance, it blows up the number of variables.

Our sparsification lower bound therefore follows a different route. We introduce a technical annotated version of the \textsc{List $P_4$-Coloring} problem. For this annotated problem, we prove a sparsification lower bound via cross-composition~\cite{BodlaenderJK14}, a technique from kernelization lower bounds. We give a polynomial-time algorithm that embeds a sequence of~$t^2$ instances of the \textsc{Clique} problem, on~$n$ vertices each, into a single instance~$(G',L')$ of \textsc{Annotated List $P_4$-Coloring}, on~$\Oh(t \cdot n^{\Oh(1)})$ vertices, which acts as the logical OR of the \textsc{Clique} inputs: there is a list coloring if and only if at least one \textsc{Clique} instance has a solution. The fact that the information from~$t^2$ distinct inputs is packed into a single instance of~$\Oh(t \cdot n^{\Oh(1)})$ vertices, means that the embedding is very efficient: the~$t^2$ $n$-vertex instances of \textsc{Clique} carry~$t^2 \cdot n^2$ bits of information (for each instance, which edges are present?), while~$G'$ has~$t^2 \cdot n^{\Oh(1)}$ potential edges, and therefore carries~$t^2 \cdot n^{\Oh(1)}$ bits of information. Applying this reduction for~$t$ a polynomial in~$n$ whose degree depends on the constant in~$n^{\Oh(1)}$, this intuitively implies that~$G'$ cannot be sparsified without losing information. Via the framework of cross-composition~\cite{BodlaenderJK14} we get the formal result that \textsc{Annotated List $P_4$-Coloring} parameterized by the number of vertices~$n$ does not admit a generalized kernelization of size~$\Oh(n^{2-\varepsilon})$ for any~$\varepsilon > 0$ unless \containment.

To transfer the lower bound for \textsc{Annotated List $P_4$-Coloring} to \textsc{List $H$-Coloring} for all graphs~$H$ which are not bi-arc, we first use a reduction inspired by Feder, Hell, and Huang~\cite{FederHH03}, to reduce to the case of bipartite graphs~$H$. Then we investigate the common structure of simple  bipartite non-bi-arc graphs~$H$, which are known to be the simple bipartite graphs~$H$ whose complement is not a circular-arc graph~\cite{FederHH03}. We uncover a common structure of such graphs which can be used to prove the incompressibility of the related \textsc{List $H$-Coloring} problems: we prove all such graphs~$H$ contain five vertices~$(a,b,c,d,e)$ such that~$H[\{a,b,c,d\}]$ is an induced~$P_4$, the open neighborhoods~$N_H(a), N_H(c)$, and~$N_H(e)$ are incomparable (i.e., none of them is contained in another), and such that also the open neighborhoods~$N_H(b), N_H(d)$ are incomparable. This 5-tuple in a bipartite graph~$H$ is sufficient to prove hardness of sparsification, which we consider one of the main contributions of the paper: We prove that the 5-tuple can be used to implement certain gadgets to enforce pairs of vertices to receive different colors in \textsc{List $H$-Coloring}. By applying these gadgets sparingly --- and not for all edges --- we reduce \textsc{Annotated List $P_4$-Coloring} to \textsc{List $H$-Coloring} without blowing up the number of vertices, and obtain Theorem~\ref{theorem:lowerbound}.

\subparagraph*{Related work}
More background on homomorphisms and \textsc{$H$-Coloring} can be found in the textbook by Hell and Ne\v{s}et\v{r}il~\cite{HellN04}, or the survey by Hahn and Tardif~\cite{HahnT97}. The classical complexity of \textsc{$H$-Coloring} has also been investigated when restricted to planar~\cite{MacGillivrayS09}, minor-closed~\cite{EsperetMOP13}, and bounded-degree~\cite{GalluccioHN00,Siggers09} input graphs~$G$.
The complexity of \listHcoloring was investigated for bounded-degree graphs~\cite{FEDER2007386}. There is also an interesting line of research concerning the descriptive and space complexity~\cite{DBLP:conf/lics/DalmauEHLR15,DBLP:conf/soda/EgriHLR14,DBLP:journals/corr/abs-0912-3802}.
Finally, the fine-grained complexity of both variants was also investigated~\cite{DBLP:conf/stacs/EgriMR18,DBLP:journals/dam/GroenlandORSSS19,Okrasa20,DBLP:conf/soda/OkrasaR20,DBLP:journals/jcss/OkrasaR20}.

\subparagraph*{Organization} Section~\ref{sec:preliminaries} contains preliminaries on kernelization and graphs. In Section~\ref{sec:annotated:lb} we present a sparsification lower bound for an annotated version of \textsc{List $P_4$-coloring}, which forms the keystone of our hardness results. In Section~\ref{sec:gadgets} we analyze the structure of hard graphs~$H$, and use that structure to build certain gadgets. These allow us to reduce the annotated problem to standard \textsc{List $H$-Coloring} problems and prove Theorem~\ref{theorem:lowerbound}.


\section{Preliminaries}
\label{sec:preliminaries}

To denote the set of numbers $1$ to $n$, we use the following notation: $[n] := \{1,\ldots,n\}$. For a set $S$ we use the notation $\binom{S}{k} := \{S' \subseteq S \mid |S'| = k\}$ to denote the set of all size-$k$ subsets of~$S$, and we define $2^{S} := \bigcup_{k=0}^{|S|} \binom{S}{k}$.
We use the notation $S^k := \{(s_1,\ldots,s_k) \mid s_1,\ldots,s_k \in S\}$ to denote the set of all $k$-tuples with elements from $S$. In particular, $[n]^2$ denotes all $2$-tuples of elements from $[n]$.

\subparagraph*{Graphs.}
All graphs considered in this paper are finite and undirected, and do not have parallel edges. We allow self-loops, unless explicitly stated otherwise.
The vertex set and the edge set of $G$ are denoted by $V(G)$ and $E(G)$, respectively. 
An edge $\{u,v\} \in E(G)$ is denoted shortly by $uv$, and by $vv$ we denote the loop on the vertex $v$.
For $v \in V(G)$, by $N_G(v)$ we denote the \emph{open neighborhood} of $v$, i.e., the set $\{u \mid uv \in E(G)\}$. The closed neighborhood of~$v$ is~$N_G[v] := N_G(v) \cup \{v\}$.
For $S\subseteq V(G)$, by $G[S]$ we denote the subgraph of $G$ induced by $S$.
%
A \emph{proper $q$-coloring} of $G$ is a function $f \colon V(G) \rightarrow [q]$ such that $f(u) \neq f(v)$  for all $uv \in E(G)$.
Let~$G$ and~$H$ be graphs. We say that~$G$ is \emph{$H$-colorable} if there exists a function $f \colon V(G) \rightarrow V(H)$ such that for all $uv \in E(G)$ it holds that $f(u)f(v) \in E(H)$. Such a function is also called a \emph{homomorphism} from~$G$ to~$H$. Note that a graph~$G$ has a homomorphism to the complete graph~$K_q$ if and only if~$G$ is (properly) $q$-colorable.
If $f$ is a homomorphism from $G$ to $H$, then we denote it by $f \colon G \to H$. We write $G \to H$ to indicate that some homomorphism from $G$ to $H$ exists.
For a graph $G$ and \emph{lists} $L \colon V(G) \to 2^{V(H)}$, a \emph{list homomorphism} from $(G,L)$ to $H$ is a homomorphism $f \colon G \to H$, such that for every $v \in V(G)$ it holds that $f(v) \in L(v)$. We write $f \colon (G,L) \to H$ if $f$ is a list homomorphism from $(G,L)$ to $H$, and $(G,L) \to H$ if some $f \colon (G,L) \to H$ exists.

\subparagraph*{Parameterized complexity.}
A \emph{parameterized problem} \Q is a subset of $\Sigma^* \times \mathbb{N}$, where $\Sigma$ is a finite alphabet.

\begin{definition}[Generalized kernel~\cite{BodlaenderJK14}] \label{def:generalized:kernel}
Let $\Q,\Q' \subseteq \Sigma^*\times\mathbb{N}$ be parameterized problems and let $h\colon\mathbb{N}\rightarrow\mathbb{N}$  be a computable function. A \emph{generalized kernel for \Q into $\Q'$ of size $h(k)$} is an algorithm that, on input $(x,k) \in \Sigma^*\times\mathbb{N}$, takes time polynomial in $|x|+k$ and outputs an instance $(x',k')$ such that:
(i)~$|x'|$ and $k'$ are bounded by $h(k)$, and
(ii)~$(x',k')\in\Q'$ if and only if $(x,k) \in\Q$.
A generalized kernel is a \emph{kernel} for $\Q$ if $\Q = \Q'$.
\end{definition}
%
%
In our applications, the complexity parameter~$k$ will be the number of vertices~$n$. We will use the framework of cross-composition, introduced by Bodlaender, Jansen, and Kratsch~\cite{BodlaenderJK14}, to establish kernelization lower bounds.

\begin{definition}[{Polynomial equivalence relation, \cite[Def. 3.1]{BodlaenderJK14}}] \label{definition:eqvr}
An equivalence relation~\eqvr on~$\Sigma^*$ is called a \emph{polynomial equivalence relation} if the following conditions hold.
\begin{itemize}
\item There is an algorithm that, given two strings~$x,y \in \Sigma^*$, decides whether~$x$ and~$y$ belong to the same equivalence class in time polynomial in~$|x| + |y|$.
\item For any finite set~$S \subseteq \Sigma^*$ the equivalence relation~$\eqvr$ partitions the elements of~$S$ into a number of classes that is polynomially bounded in the size of the largest element of~$S$.\qedhere
\end{itemize}
\end{definition}

\begin{definition}[{Cross-composition, \cite[Def. 3.7]{BodlaenderJK14}}]\label{definition:crosscomposition}
Let~$L\subseteq\Sigma^*$ be a language, let~$\eqvr$ be a polynomial equivalence relation on~$\Sigma^*$, let~$\Q\subseteq\Sigma^*\times\N$ be a parameterized problem, and let~$f \colon \N \to \N$ be a function. An \emph{\OR-cross-com\-position of~$L$ into~$\Q$} (with respect to \eqvr) \emph{of cost~$f(t)$} is an algorithm that, given~$t$ instances~$x_1, x_2, \ldots, x_t \in \Sigma^*$ of~$L$ belonging to the same equivalence class of~$\eqvr$, takes time polynomial in~$\sum _{i=1}^t |x_i|$ and outputs an instance~$(y,k) \in \Sigma^* \times \mathbb{N}$ such that:
\begin{itemize}
\item The parameter~$k$ is bounded by $\Oh(f(t)\cdot(\max_i|x_i|)^c)$, where~$c$ is some constant independent of~$t$, and
\item instance $(y,k) \in \Q$ if and only if there is an~$i \in [t]$ such that~$x_i \in L$.\label{property:OR}\qedhere
\end{itemize}
\end{definition}

\begin{theorem}[{\cite[Theorem 3.8]{BodlaenderJK14}}] \label{thm:cross_composition_LB}
Let~$L\subseteq\Sigma^*$ be a language, let~$\Q\subseteq\Sigma^*\times\N$ be a parameterized problem, and let~$d,\varepsilon$ be positive reals. If~$L$ is NP-hard under Karp reductions, has an \OR-cross-composition into~$\Q$ with cost~$f(t)=t^{1/d+o(1)}$, where~$t$ denotes the number of instances, and~$\Q$ has a polynomial (generalized) kernelization with size bound~$\Oh(k^{d-\varepsilon})$, then \containment.
\end{theorem}

We will refer to an \OR-cross-composition of cost~$f(t) = \sqrt{t} \log (t)$ as a \emph{degree-$2$ cross-composition}. By Theorem~\ref{thm:cross_composition_LB}, a degree-$2$ cross-composition can be used to rule out generalized kernels of size~$\Oh(k^{2 - \varepsilon})$ and thus provides a way to obtain sparsification lower bounds. 
Generalized kernelization lower bounds can be transferred using the notion of linear-parameter transformations.

\begin{definition}[Linear-parameter transformation]\label{def:lpt}
Let $\mathcal{P}, \mathcal{Q} \subseteq \Sigma^* \times \mathbb{N}$ be two parameterized problems.
A \emph{linear-parameter transformation} from
$\mathcal{P}$
to $\mathcal{Q}$ is a polynomial-time algorithm that, given an instance $(x,k) \in \Sigma^* \times \mathbb{N}$ of $\mathcal{P}$, outputs an instance $(x',k') \in \Sigma^* \times \mathbb{N}$ of $\mathcal{Q}$ such that the following holds:
(i) $(x,k) \in \mathcal{P}$ if and only if $(x',k') \in \mathcal{Q}$, and
(ii) $k' \in \mathcal{O}(k)$.
\end{definition}

\noindent It is well-known~\cite{BodlaenderJK14} that the existence of a linear-parameter transformation from problem~$\mathcal{P}$ to~$\mathcal{Q}$ implies that any generalized kernelization lower bound for~$\mathcal{P}$, also holds for~$\mathcal{Q}$.

\section{Lower bound for Annotated List $P_4$-Coloring} \label{sec:annotated:lb}
We prove a sparsification lower bound for the following problem, where we take $P_4$ to be the graph on vertices $\{\cola,\colb,\colc,\cold\}$ with edges $\cola \colb, \colb\colc,\colc\cold$.

\defproblem{\Anlistcoloring}
{A tuple $(G,L,\mathcal{S},F)$, such that $G$ is a simple undirected bipartite graph with bipartition~$V(G) = V_1 \cup V_2$, $L \colon V(G) \to 2^{\{\cola,\colb,\colc,\cold\}}$ with~$L(v) \subseteq \{\cola,\colc\}$ for all $v\in V_1$ and~$L(v) \subseteq \{\colb,\cold\}$ for all $v\in V_2$, $\mathcal{S} = S_1, \ldots, S_m$ is a sequence such that~$S_i \subseteq V_1$  for each~$i \in [m]$ satisfying~$\sum _{i=1}^m |S_i| \leq 3|V(G)|$, and $F \subseteq \binom{V_1}{2} \cup \binom{V_2}{2}$ is a set with~$|F| \leq |V(G)|$.}
{Does~$G$ admit a homomorphism~$f \colon V(G) \to \{\cola,\colb,\colc,\cold\}$ to the graph~$P_4$  with~$f(v) \in L(v)$ for all~$v \in V(G)$, such that for all~$i \in [m]$ there is a vertex~$v \in S_i$ with~$f(v) \neq \colc$, and such that for all~$\{u,v\} \in F$ we have~$f(u) \neq f(v)$?}

Intuitively, the annotations allow one to express two types of additional constraints on the coloring~$f$. Using a set~$S_i$, one can enforce that at least one vertex is not colored~$\colc$. Using a pair~$\{u,v\} \in F$, one can enforce that~$u$ and~$v$ do not receive the same color. While the latter can easily be expressed by simply inserting an edge between~$u$ and~$v$ in a \textsc{$K_q$-Coloring} instance, this needs a nontrivial gadget for general graphs~$H$.

\begin{lemma}\label{lem:annotated-coloring:LB}
\Anlistcoloring parameterized by the number of vertices~$n$ admits no generalized kernel of size~$\Oh(n^{2-\varepsilon})$ for any~$\varepsilon > 0$, unless \containment.
\end{lemma}

\newcommand{\col}{\ensuremath{h}\xspace}%
\newcommand{\coloring}{annotated $P_4$-coloring\xspace}%
\newcommand{\acoloring}{an annotated $P_4$-coloring\xspace}%
\newcommand{\colorable}{annotated $P_4$-colorable\xspace}%
\begin{proof}
We will prove this lower bound by giving a degree-$2$ cross-composition from \textsc{Clique} to \Anlistcoloring.  We define a polynomial equivalence relation $\mathcal{R}$ on instances of \textsc{Clique}. Let any two instances that ask for a clique that is larger than their respective number of vertices be equivalent; these are always no-instances. Let two instances of \textsc{Clique} be equivalent under $\mathcal{R}$, when the input graphs have same number of vertices and the problems ask for a clique of the same size. It is easy to verify that $\mathcal{R}$ is indeed a polynomial equivalence relation.

By duplicating one of the inputs multiple times as needed, we can assume the number of inputs to the cross-composition is a square. Therefore, assume we are given $t$ instances of \textsc{Clique}, such that $t' := \sqrt{t}$ is integer and such that each instance has $n$ vertices and asks for a size-$k$ clique. Enumerate the given input instances as $X_{i,j}$ for $i,j\in[t']$ and let $G_{i,j}$ denote the corresponding graph. Label the vertices in each instance arbitrarily as $x_1,\ldots,x_n$. 
We show how to create an instance $(G,L,\mathcal{S},F)$ that is a yes-instance for \Anlistcoloring
 if and only if at least one of the given instances for \textsc{Clique} is a yes-instance. Refer to  Figure~\ref{fig:cross-composition} for a sketch.

\begin{figure}[t]
 \includegraphics{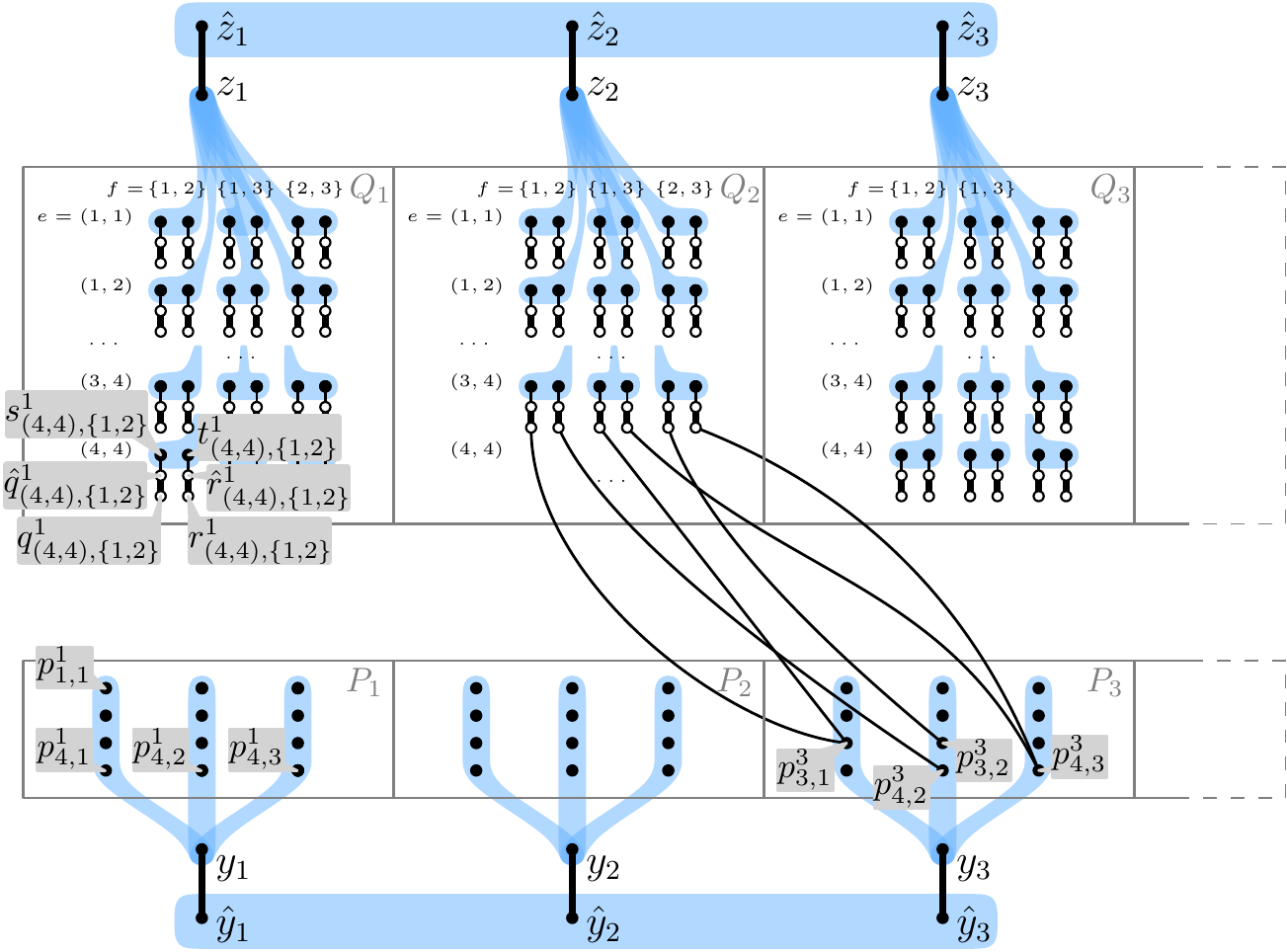}
 \caption{A sketch of the created graph $G$, for $n = 4$, and $k = 3$ where  $x_3x_4 \notin E(G_{2,3})$. Edges between $P$ and $Q$ are omitted, except for the edges that result from the fact that $x_3x_4 \notin E(G_{2,3})$. A fat edge between $u$ and $v$ indicates that $\{u,v\} \in F$. Vertex sets contained in $\mathcal{S}$ are marked in blue. White vertices have lists $\{\colb,\cold\}$ while black vertices have list $\{\cola,\colc\}$. Note that the constructed graph is bipartite with the white and black vertices as partite sets.}
 \label{fig:cross-composition}
\end{figure}

\begin{enumerate}
 \item  \label{step:vertices_bottom} For each $j \in [t']$, $\ell \in [n]$, and $m \in [k]$ create a vertex $p^j_{\ell,m}$. Let $\ELL(p^j_{\ell,m}) := \{\cola,\colc\}$. Let $P_j$ contain all created vertices $p^j_{\ell,m}$ for $\ell \in [n]$, $m \in [k]$. Let $P:= \bigcup_{j \in [t']} P_j$. 
 \item \label{step:vertices_top} 
 For each $f \in \binom{[k]}{2}$, each $e = (e_1,e_2) \in [n]^2$, and each $i \in [t']$, create vertices $q^i_{e,f}$, $r^i_{e,f}$, $\hat{q}^i_{e,f}$, $\hat{r}^i_{e,f}$, $s^i_{e,f}$, and $t^i_{e,f}$. Let $Q_i :=\{q^i_{e,f},r^i_{e,f},\hat{q}^i_{e,f},\hat{r}^i_{e,f},s^i_{e,f},t^i_{e,f}\mid  f \in \binom{[k]}{2}, e \in [n]^2\}$. Note that $Q_i$ contains $6\binom{k}{2}$ vertices for each ordered pair of vertices in an $n$-vertex graph; these pairs model edges and self-loops. Let $Q:= \bigcup_{i \in [t']} Q_i$. Now let $\ELL(q^i_{e,f}):= \ELL(r^i_{e,f}) := \ELL(\hat{q}^i_{e,f}) := \ELL(\hat{r}^i_{e,f}):=\{\colb,\cold\}$ and $\ELL(s^i_{e,f}) := \ELL(t^i_{e,f}) := \{\cola,\colc\}$.
 \item \label{step:connect_vertices_top} For each  $f \in \binom{[k]}{2}$, each $e = (e_1,e_2) \in [n]^2$, and each $i \in [t']$, do the following. Connect vertex $\hat{q}^i_{e,f}$ to vertex $s^i_{e,f}$, and connect vertex $\hat{r}^i_{e,f}$ to vertex $t^i_{e,f}$. This ensures that when $\hat{q}^i_{e,f}$ (respectively, $\hat{r}^i_{e,f}$) gets color $\cold$, then $s^i_{e,f}$ (respectively $t^i_{e,f}$) always gets color $\colc$, since vertex~$\colc$ is the unique neighbor of vertex~$\cold$ in~$P_4$. If however $\hat{q}^i_{e,f}$ gets color $\colb$, then $s^i_{e,f}$ can receive color $\cola$ or $\colc$. 
  Add the pairs $\{q^i_{e,f},\hat{q}^i_{e,f}\}$ and $\{\hat{r}^i_{e,f},r^i_{e,f}\}$ to $F$. Verify that when both $q^i_{e,f}$ and $r^i_{e,f}$ get color $\colb$, then $s^i_{e,f}$ and $t^i_{e,f}$ must get color $\colc$. 
\end{enumerate}
Recall that the goal of the construction is to ensure that the \Anlistcoloring instance~$(G, \ELL, \mathcal{S}, F )$ acts as the logical \textsc{or} of the \textsc{Clique} instances~$X_{i,j}$, so that~$G$ has a~coloring respecting the lists and annotations if and only if some input graph~$G_{i,j}$ has a clique of size~$k$. The part of~$G$ constructed so far allows colorings of~$G$ to encode the vertex set of a $k$-clique through its behavior on $P$. Finding a proper list coloring of~$G$ entails highlighting vertices from one set $P_j$ that correspond to a clique in instance $X_{i,j}$ for some $i \in [t']$. The highlighting property will be enforced by ensuring at least one vertex in each set $\{p^j_{\ell,m}\mid \ell \in [n]\}$ for~$m \in [k]$ receives color $\cola$. The index of the vertex that is colored~$\cola$ encodes the $m$-th vertex in the clique to which the coloring corresponds. The vertices in $Q_i$ are then used to verify that the selected vertices form a clique in $G_{i,j}$. The next steps add additional vertices and edges, in order to achieve these properties.
\begin{enumerate}
\setcounter{enumi}{3}
\item \label{step:connect_top_bottom} For each $i,j \in [t']$, consider instance $X_{i,j}$. For all $f \in \binom{[k]}{2}$ and $e = (e_1,e_2) \in [n]^2$, connect vertex $p^j_{e_1,f_1}$ to $q^i_{e,f}$ and connect $p^j_{e_2,f_2}$ to $r^i_{e,f}$ whenever $x_{e_1}x_{e_2} \notin E(G_{i,j})$. Here $f_1 < f_2$ are such that $f = \{f_1,f_2\}$. Observe that in particular (since $G_{i,j}$ is a simple graph), we have that $x_{e_1}x_{e_1}\notin E(G_{i,j})$ for all $e_1 \in [n]$. Observe also that each vertex~$q^i_{e,f}, r^i_{e,f}$ has a \emph{unique} neighbor in~$P_j$ for each~$j \in [t']$.
\end{enumerate}
The above step will allow using the coloring of vertices $s^i_{e,f}$ and $t^i_{e,f}$ to verify that the vertices selected in $P_j$ correspond to a clique: when $x_{e_1}x_{e_2}$ is not an edge, they will ensure that we cannot select both.
\begin{enumerate}
\setcounter{enumi}{4}
 \item \label{step:add_y} Add vertices $y_j$ and $\hat{y}_j$ for all $j \in [t']$ and let $Y:= \{y_j \mid j \in [t']\}$, $\hat{Y} := \{\hat{y}_j\mid j \in [t']\}$. Let $\ELL(y_j) := \ELL(\hat{y}_j) := \{\cola,\colc\}$ for all $j \in [t']$. 
 \item \label{step:add_z} Similarly, add vertices $z_i$ and $\hat{z}_i$ for all $i \in [t']$ and let $Z := \{z_i \mid i \in [t']\}$, $\hat{Z} := \{\hat{z}_i\mid i \in [t']\}$. Let $\ELL(z_i) := \ELL(\hat{z}_i) := \{\cola,\colc\}$.
 \item \label{step:gadget_yz} Add the sets $\hat{Y}$ and $\hat{Z}$ to $\mathcal{S}$. Furthermore, for all $i \in [t']$, add $\{y_i,\hat{y}_i\}$ and $\{z_i,\hat{z}_i\}$ to $F$. 
\end{enumerate}
The steps above ensure that at least one vertex $y_j \in Y$ receives color $\colc$ and at least one vertex in $z_i \in Z$ receives color $\colc$. This will indicate that instance $X_{i,j}$ is selected. We will now put further constraints on the coloring of $P_j$ and $Q_i$ when they correspond to a selected instance.
\begin{enumerate}
 \setcounter{enumi}{7}
 \item \label{step:gadget_bottom} For all $j \in [t']$, $m \in [k]$, we add the set $\{y_j\} \cup \{p^j_{\ell,m} \mid \ell \in [n]\}$ to $\mathcal{S}$.  
 \item \label{step:gadget_top} For all $i \in [t']$, for all $f \in \binom{[k]}{2}$ and $e \in [n]^2$, add the set $\{s^i_{e,f},t^i_{e,f},z_i\}$ to $\mathcal{S}$.
\end{enumerate}
This concludes the construction of $G$, $\ELL$, $\mathcal{S}$ and $F$. 
Let us start by counting the number of vertices in $G$:
  \begin{align*} |V(G)| = \underbrace{t'\cdot n\cdot k}_{|P|} + \underbrace{t'\cdot (n^2 \cdot \binom{k}{2} \cdot 6)}_{|Q|} + \underbrace{t' + t' + t' + t'}_{|Y| + |\hat{Y}| + |Z| + |\hat{Z}|} = \Oh(\sqrt{t}\cdot n^2\cdot k^2). \end{align*}
Observe that hereby $|V(G)|$ is properly bounded for a degree-$2$ cross composition. 

We continue by showing that $G$ is a valid instance of \Anlistcoloring.
Verify  that $G$ is bipartite with bipartition $V_1 = P \cup Y \cup \hat{Y}\cup Z\cup \hat{Z} \cup \{s^i_{e,f},t^i_{e,f} \mid f \in \binom{[k]}{2}, e\in[n]^2, i \in [t']\}$ and $V_2 = \{q^i_{e,f},r^i_{e,f},\hat{q}^i_{e,f},\hat{r}^i_{e,f} \mid f \in \binom{[k]}{2}, e\in[n]^2, i \in [t']\}$. Hence, $V_1$ contains all vertices whose lists are a subset of $\{\cola,\colc\}$ and $V_2$ contains all remaining vertices, and it can be verified that the lists of these vertices are a subset of $\{\colb,\cold\}$. Observe that indeed each set in $F$ is a subset of either $V_1$ or $V_2$, and each set in $\mathcal{S}$ is a subset of $V_1$.

Furthermore, it is straightforward to verify that $|F| \leq |V(G)|$ as promised for \Anlistcoloring (note that we only add elements to $F$  in Steps~\ref{step:connect_vertices_top} and~\ref{step:gadget_yz}). We can also verify that 
\begin{align*}
\sum_{S \in \mathcal{S}} |S| \leq \underbrace{2\cdot t'}_{\text{Step~\ref{step:gadget_yz}}} + \underbrace{t'\cdot k \cdot (n+1)}_{\text{Step~\ref{step:gadget_bottom}}}+\underbrace{t'\cdot n^2 \cdot \binom{k}{2}\cdot 3}_{\text{Step~\ref{step:gadget_top}}} \leq 3|V(G)|.
\end{align*}
As such, we have created a valid instance of \Anlistcoloring.
The next two claims show that the constructed graph $G$ indeed acts as the logical \textsc{or} of the given input instances.

\begin{claim}
\label{claim:cc:correctness-1}
 If some input graph $G_{i^*,j^*}$ has a clique of size $k$, then $G$ is \colorable.
\end{claim}
\begin{claimproof}
Let such $i^*,j^* \in [t']$ be given, we create a \acoloring $\col \colon V(G) \to \{\cola,\colb,\colc,\cold\}$ for $G$. First of all, for all $j \neq j^*$ with $j\in [t']$, let $\col(y_j) := \cola$ and let $\col(\hat{y}_j) := \colc$. Let $\col(y_{j^*}) := \colc$ and let $\col(\hat{y}_{j^*}) := \cola$. Similarly, for $i \neq i^*$ we let $\col(z_i) := \cola$ and let $\col(\hat{z}_i) := \colc$. Furthermore define $\col(z_{i^*}) := \colc$ and $\col(\hat{z}_{i^*}) := \cola$. Hereby, not all vertices in $\hat{Y}$ have color $\colc$, and not all vertices in $\hat{Z}$ have color $\colc$, such that we satisfy the sets added to $\mathcal{S}$ in Step~\ref{step:gadget_yz} of the construction.

For all $p \in P_j$ for $j \neq j^*$, let $\col(p) := \colc$. Furthermore, for all $e \in [n]^2$, $f \in \binom{[k]}{2}$ and $i \neq i^*$ with $i \in [t']$, we define $\col(q^i_{e,f}) := \col(r^i_{e,f}) = \colb$, $\col(\hat{q}^i_{e,f}) := \col(\hat{r}^i_{e,f}) = \cold$, and $\col(s^i_{e,f}) := \col(t^i_{e,f}) = \colc$. 

It remains to color the vertices in $P_{j^*}$ and $Q_{i^*}$.  Let $K = \{x_{i_1},\ldots,x_{i_k}\}$ be a clique in $G_{i^*,j^*}$ of size $k$. For $m \in [k]$, $\ell \in [n]$ let $\col(p^{j^*}_{\ell,m}) := \cola$ if $i_m = \ell$. Otherwise, let 
$\col(p^{j^*}_{\ell,m}) := \colc$. In this way, for each $m \in [k]$, the set $\{y_{j^*}\} \cup \{p^{j^*}_{\ell,m} \mid \ell \in [n]\}$ contains a vertex that receives color $\cola$, as desired. We now extend this coloring to $Q_{i^*}$. Let $e = (e_1,e_2) \in [n]^2$ and let $f \in \binom{[k]}{2}$ such that $f = \{f_1,f_2\}$ for $f_1 < f_2$. Let $\col(q^{i^*}_{e,f}) := \colb$ if the unique neighbor of $q^{i^*}_{e,f}$ in $P_{j^*}$ has color $\cola$. Otherwise, let  $\col(q^{i^*}_{e,f}) := \cold$. We color  $r^{i^*}_{e,f}$ in the same way, thus  $\col(r^{i^*}_{e,f}) := \colb$ if its unique neighbor in~$P_{j^*}$ has color $\cola$, and  $\col(r^{i^*}_{e,f}) := \cold$ otherwise. Color  $\hat{q}^{i^*}_{e,f}$ with the only color in $\{\colb,\cold\} \setminus \{\col(q^{i^*}_{e,f})\}$ and similarly color $\hat{r}^{i^*}_{e,f}$ with the only color in $\{\colb,\cold\} \setminus \{\col(r^{i^*}_{e,f})\}$. Finally, let $\col(s^{i^*}_{e,f}) := \colc$ if $\col(\hat{q}^{i^*}_{e,f}) = \cold$ and let $\col(s^{i^*}_{e,f}) := \cola$ otherwise. Similarly, let $\col(t^{i^*}_{e,f}) := \colc$ if $\col(\hat{r}^{i^*}_{e,f}) = \cold$ and let $\col(t^{i^*}_{e,f}) := \cola$ otherwise.
This concludes the definition of \col. It remains to show that \col is a valid \coloring of $G$.
We split this into three parts.

First of all, we verify that each $S \in \mathcal{S}$ contains a vertex that does not get color $\colc$. For $\hat{Y}$ and $\hat{Z}$ this was verified before. Consider a set $\{y_j\} \cup \{p^j_{\ell,m} \mid \ell \in [n]\}$ added in Step~\ref{step:gadget_bottom}. Observe that if $j \neq j^*$ then $y_j$ has color $\cola$ and we are done. Otherwise, by definition, we have $\col(p^{j^*}_{i_m,m}) := \cola$ and thus indeed this set has a vertex of color $\cola$. Now consider a set $\{s^i_{e,f},t^i_{e,f},z_i\}$ added in Step~\ref{step:gadget_top}. If $i \neq i^*$, vertex $z_i$ has color $\cola$ and we are done. Otherwise if $i = i^*$, we claim that it cannot be the case that $\col(s^{i^*}_{e,f}) = \col(t^{i^*}_{e,f}) = \colc$. Suppose towards a contradiction that indeed both these vertices have color $\colc$. By the choice of our coloring, this implies that $\col(\hat{q}^{i^*}_{e,f}) = \col(\hat{r}^{i^*}_{e,f}) = \cold$ and thus $\col(q^{i^*}_{e,f}) = \col(r^{i^*}_{e,f}) = \colb$. Letting~$e = (e_1, e_2) \in [n]^2$ and~$f = \{f_1, f_2\}$ for~$f_1 < f_2$, that means that $q^{i^*}_{e,f}$ and $r^{i^*}_{e,f}$ have their unique neighbor in $P_{j^*}$  of color $\cola$, implying $\col(p^{j^*}_{e_1,f_1}) = \col(p^{j^*}_{e_2,f_2}) = \cola$. So these edges were constructed in Step~\ref{step:connect_top_bottom}, implying $x_{e_1}x_{e_2} \notin E(G_{i^*,j^*})$. Since $x_{e_1}\in K$ and $x_{e_2} \in K$, this contradicts that $K$ is a clique.

Secondly, verify that for all pairs in $\{u,v\} \in F$, $\col(u) \neq \col(v)$: we only add sets to $F$ in Steps~\ref{step:connect_vertices_top} and~\ref{step:gadget_yz}. We always ensure in the construction that if $\{u,v\} \in F$, the two vertices get different colors.

Thirdly, we verify the coloring of endpoints of edges in $G$. First of all, consider the edges added in Step~\ref{step:connect_vertices_top} and observe that we always color the endpoints properly in the description above: if $\hat{q}^i_{e,f}$ gets color $\cold$, we color $s^i_{e,f}$ with $\colc$ which is allowed; if $\hat{q}^i_{e,f}$ has color $\colb$, we use color $\cola$ in $s^i_{e,f}$ which is again fine. One may verify that the same holds for edges $\hat{r}^i_{e,f}t^i_{e,f}$. Now consider the edges between a vertex $u \in P$ and $v \in Q$. If $u \notin P_{i^*}$ it follows that $\col(u) = \colc$. Since by the lists, $\col(v) \in \{\colb,\cold\}$ this implies that this edge is properly colored. Similarly, if $v \notin Q_{j^*}$ we obtain $\col(v) = \colb$ and since $\col(u) \in \{\cola,\colc\}$ we are again done. If $u \in P_{j^*}$ and $v \in P_{i^*}$ one may observe that the edge $uv$ is properly colored by definition: $v$ has color $\cold$ only if it has no neighbors of color $\cola$ (and $\col(u) \in \{\cola,\colc\}$ thus implies $\col(u) = \colc$), and otherwise $v$ has color $\colb$ such that the edge is again properly colored by $\col(u) \in \{\cola,\colc\}$. 
\end{claimproof}

\begin{claim}\label{claim:cc:correctness-2}
 If $G$ has \acoloring \col, then there exist $i^*,j^* \in [t']$ such that $G_{i^*,j^*}$ has a clique of size $k$.
\end{claim}
\begin{claimproof}
 Since $\hat{Y},\hat{Z} \in \mathcal{S}$, there exist $i^*,j^* \in [t']$ such that $\col(\hat{y}_{j^*}) \neq \colc$ and $\col(\hat{z}_{i^*}) \neq \colc$, implying by the lists that $\col(\hat{y}_{j^*})=\col(\hat{z}_{i^*})=\cola$. Since $\{\hat{y}_{j^*},y_{j^*}\} \in F$ and $\{\hat{z}_{i^*},z_{i^*}\} \in F$ (by Step~\ref{step:gadget_yz}) we obtain that $\col(y_{j^*})=\col(z_{i^*})=\colc$. Now since $\{y_{j^*}\} \cup \{p^{j^*}_{\ell,m} \mid \ell \in [n]\} \in \mathcal{S}$ for all $m\in[k]$, it follows that for all $m \in [k]$, there exists $i_m \in [n]$ such that $\col(p^{j^*}_{i_m,m}) = \cola$. Let $x_1,\ldots,x_n$ be the vertices of $G_{i^*,j^*}$, define $K:= \{x_{i_1},\ldots,x_{i_k}\}$. We show that $K$ is a size-$k$ clique in $G_{i^*,j^*}$ by showing that $x_{i_m}x_{i_{m'}}$ is an edge for all $m \neq m'$. Observe that this then also proves that all selected vertices are distinct as the input graphs have no self-loops.
 
 \sloppy
 Let $m,m' \in [k]$. Without loss of generality let $m < m'$. Suppose towards a contradiction that $x_{i_{m}}x_{i_{m'}} \notin E(G_{i^*,j^*})$. Then, in Step~\ref{step:connect_top_bottom}, we added the edges $p^{j^*}_{i_m,m} q^{i^*}_{(i_m,i_{m'}), \{m,m'\}}$ and $p^{j^*}_{i_{m'},m'} r^{i^*}_{(i_m,i_{m'}), \{m,m'\}}$. Note that since we choose $x_{i_m},x_{i_{m'}} \in K$, it must hold that $\col(p^{j^*}_{i_m,m}) = \col(p^{j^*}_{i_{m'},m'}) = \cola$. Since $\colb$ is the only neighbor of $\cola$ in the $P_4$, we get $\col(q^{i^*}_{(i_m,i_{m'}), \{m,m'\}}) = \col(r^{i^*}_{(i_m,i_{m'}), \{m,m'\}}) = \colb$. Since in Step~\ref{step:connect_vertices_top} we added $\{q^{i^*}_{(i_m,i_{m'}), \{m,m'\}},\allowbreak \hat{q}^{i^*}_{(i_m,i_{m'}), \{m,m'\}}\}$ and $\{r^{i^*}_{(i_m,i_{m'}), \{m,m'\}},\hat{r}^{i^*}_{(i_m,i_{m'}), \{m,m'\}}\}$  to $F$, we obtain $\col(\hat{q}^{i^*}_{(i_m,i_{m'}), \{m,m'\}}) = \col(\hat{r}^{i^*}_{(i_m,i_{m'}), \{m,m'\}}) = \cold$. Since $\hat{q}^{i^*}_{(i_m,i_{m'}) \{m,m'\}} s^{i^*}_{(i_m,i_{m'}), \{m,m'\}}$ and $\hat{r}^{i^*}_{(i_m,i_{m'}) \{m,m'\}}\allowbreak t^{i^*}_{(i_m,i_{m'}), \{m,m'\}}$ are edges in $G$ (also added in Step~\ref{step:connect_vertices_top}), we get that $\col(s^{i^*}_{(i_m,i_{m'}), \{m,m'\}}) = \col(t^{i^*}_{(i_m,i_{m'}), \{m,m'\}}) = \colc$. However, note that $\{r^{i^*}_{(i_m,i_{m'}), \{m,m'\}},r^{i^*}_{(i_m,i_{m'}), \{m,m'\}},z_{i^*}\} \in \mathcal{S}$, by Step~\ref{step:gadget_top}. These three vertices all have color $\colc$, contradicting that \col is a valid \coloring of $G$.
 \fussy 
\end{claimproof}
Using the claims above and the bound on the size of $V(G)$ computed earlier, we conclude that we have given a degree-$2$ cross-composition to \coloring, such that the lower bound follows from \cref{thm:cross_composition_LB}.
\end{proof}

\section{Gadgets in hard graphs for \listcoloring{$H$}} \label{sec:gadgets}

Now we are going back to investigating the \listHcoloring problem, for fixed graphs $H$. To transfer the lower bound of \cref{lem:annotated-coloring:LB} to \listHcoloring for all graphs~$H$ which are not bi-arc graphs, we use a two-step process. First we use an idea of Feder, Hell, and Huang~\cite{FederHH03} which allows us to efficiently reduce so-called \emph{consistent} instances of the \listcoloring{$H^*$} problem, where~$H^*$ is a (simple) bipartite graph naturally associated to~$H$, to equivalent instances of \listHcoloring on the same vertex set. This implies that \listHcoloring is at least as hard to sparsify as consistent instances \listcoloring{$H^*$}, where~$H^*$ is a bipartite graph. Then we will develop a number of gadgets to reduce \Anlistcoloring to \listcoloring{$H^*$} on consistent instances, in a way that preserves sparsification lower bounds. Together, this chain of reductions will prove \cref{theorem:lowerbound}.


\subsection{Bi-arc graphs, associated bipartite graphs, and consistent instances}\label{sec:biarc}

Recall that the complexity dichotomy for \listHcoloring was proven in three steps:
\begin{enumerate}
\item for reflexive $H$, the polynomial cases appear to be interval graphs~\cite{FederH98},
\item for irreflexive $H$, the polynomial cases appear to be bipartite co-circular-arc graphs~\cite{Feder99},
\item for general graphs, the polynomial cases are the so-called bi-arc graphs~\cite{FederHH03}.
\end{enumerate}

The main idea of showing the final step of the dichotomy was a reduction to the bipartite case. 
For a graph $H$, by $H^*$ we denote the \emph{associated bipartite graph}, defined as follows.
The vertex set of $H^*$ is the union of two independent sets: $V_1 :=\{x' \mid x \in V(H)\}$ and $V_2 :=\{x'' \mid x \in V(H)\}$. 
The vertices $x' \in V_1$ and $y'' \in V_2$ are adjacent if and only if $xy \in E$. Note that the edges of type $x'x''$ in $H^*$ correspond to loops in $H$. 

As we mentioned in the introduction, bi-arc graphs are defined in terms of certain geometric representation, but for us much more convenient will be to use the following characterization in terms of the associated bipartite graph.

\begin{theorem}[Feder, Hell, and Huang~\cite{FederHH03}]
Let $H$ be an undirected graph, possibly with loops. The following are equivalent.
\begin{enumerate}
\item $H$ is a bi-arc graph.
\item $H^*$ is the complement of a circular-arc graph.
\end{enumerate}
\end{theorem}

Thus the graphs $H$ for which \listHcoloring is NP-hard, are precisely those for which \listcoloring{$H^*$} is NP-hard: when $H^*$ is the complement of a circular-arc graph.

Now let us explain how showing the hardness of \listHcoloring can be reduced to showing the hardness of \listcoloring{$H^*$}. Here we need the notion of a \emph{consistent} instance of the problem.

\begin{definition}
Let~$F$ be a connected bipartite graph with bipartition classes~$X$ and~$Y$. An instance $(G,L)$ of \listcoloring{$F$} is \emph{consistent}, if $G$ is bipartite and has a bipartition into classes $A,B \subseteq V(G)$, such that~$L(a) \subseteq X$ for all~$a \in A$, and~$L(b) \subseteq Y$ for all~$b \in B$.
\end{definition}

The following Proposition follows from the idea of Feder, Hell, and Huang~\cite{FederHH03}, and provides a reduction from \listcoloring{$H^*$} to \listHcoloring that preserves the vertex set of~$G$. Its exact statement comes from~\cite{Okrasa20,DBLP:journals/corr/abs-2006-11155}.

\begin{proposition}[Okrasa et al.~\cite{Okrasa20,DBLP:journals/corr/abs-2006-11155}]
\label{prop:bipartite-associted} 
Let $H$ be a graph and let $(G,L)$ be a consistent instance of \listcoloring{$H^*$}.
Define $L' \colon V(G) \to 2^{V(H)}$ as $L'(x) := \{u \mid \{u',u''\} \cap L(x) \neq \emptyset\}$.
Then $(G,L) \to H^*$ if and only if $(G,L') \to H$.
\end{proposition}

\subsection{Hard bipartite graphs $H$}


The following notion was introduced by Feder, Hell, and Huang~\cite{Feder99}.

\begin{definition}\label{def:edge-ast}
Let $k\geq 1$ and let $H$ be a bipartite graph with bipartition classes $X,Y$. Let $U=\{u_0,\ldots,u_{2k}\} \subseteq X$ and $V=\{v_0,\ldots,v_{2k}\} \subseteq Y$ be ordered sets of vertices such that $\{u_0v_0, u_1v_1,\ldots,u_{2k}v_{2k}\}$ is a set of edges of $H$. 
We say that $(U,V)$ is a \emph{special edge asteroid} (or, in short, an \emph{asteroid}) of order $2k+1$, if for every $i \in \{0,\ldots,2k\}$ there exists a $u_i$-$u_{i+1}$-path $P_{i,i+1}$ in $H$ (indices are computed modulo $2k+1$), such that
\begin{enumerate}[(a)]
\item there are no edges between $\{u_i,v_i\}$ and $\{v_{i+k},v_{i+k+1}\}\cup V(P_{i+k,i+k+1})$ and
\item there are no edges between $\{u_0,v_0\}$ and $\{v_1,\ldots,v_{2k}\} \cup \bigcup_{i=1}^{2k-1} V(P_{i,i+1})$. \label{it:special}
\end{enumerate}
\end{definition}

%

Feder, Hell, and Huang showed the following characterization of hard bipartite cases of \listHcoloring, i.e., bipartite graphs $H$, whose complement is not a circular-arc graph.

\begin{theorem}[Feder et al.~\cite{Feder99}]\label{thm:NP-hard-circular-arc}
A bipartite graph $H$ is not the complement of a circular-arc graph if and only if $H$ contains an induced cycle with at least 6 vertices or an asteroid.
\end{theorem}

While induced cycles of length at least 6 and asteroids suffice to prove NP-completeness of \listHcoloring, to prove sparsification lower bounds via \Anlistcoloring we need a more local structure. We therefore introduce the following notion.

\begin{definition}
An \emph{extended $P_4$ gadget} in an undirected simple graph~$H$ is a tuple~$(a,b,c,d,e)$ of distinct vertices in~$H$, such that all of the following hold:
\begin{enumerate}
	\item $H[\{a,b,c,d\}]$ is isomorphic to~$P_4$,
	\item the sets~$N_H(a), N_H(c), N_H(e)$ are pairwise incomparable, and
	\item the sets~$N_H(b), N_H(d)$ are pairwise incomparable.
\end{enumerate}
\end{definition}

Intuitively, if~$H$ contains an extended $P_4$ gadget, then the~$P_4$ on~$(a,b,c,d)$ allows a \listHcoloring instance to express a homomorphism problem to~$P_4$, while the presence of vertex~$e$ and the incomparability of the neighborhoods allows gadgets to be constructed to enforce the semantics of the set~$F$ and the sequence~$\mathcal{S}$ in the definition of \Anlistcoloring, thereby allowing a reduction from that problem to the \listHcoloring. The gadgets needed to simulate the pairwise constraints from~$F$ are given by the next lemma.

\begin{lemma}
\label{lemma:extendedP4exists}
Let $H$ be a bipartite graph which contains an induced cycle of at least 6 vertices or an asteroid. Then there exists an extended $P_4$ gadget~$(a,b,c,d,e)$ in~$H$. Moreover, for every $Q \in \{\{a,c,e\}, \{b,d\}\}$ there is a consistent \listHcoloring instance~$(G_Q, L)$ containing two distinguished vertices~$\gamma_1, \gamma_2$ such that a mapping~$f \colon \{\gamma_1, \gamma_2\} \to Q$ can be extended to a proper list $H$-coloring of~$(G_Q,L)$ if and only if~$f(\gamma_1) \neq f(\gamma_2)$.
\end{lemma}

We remark that it is actually sufficient to show that every bipartite graph $H$ which contains an induced cycle of at least 6 vertices or an asteroid, contains also an extended $P_4$ gadget. In this situation, the existence of $(G_{\{a,c,e\}},L)$ and $(G_{\{b,d\}},L)$ follows from a result in~\cite{Okrasa20,DBLP:journals/corr/abs-2006-11155}. However, for the sake of completeness, we include the whole proof.

Before we proceed to the construction of an extended $P_4$ gadget, let us introduce some definitions.
A \emph{walk} $\cP$ is a sequence $p_1,\ldots, p_\ell$ of vertices of $H$, such that $p_ip_{i+1} \in E(H)$, for every $i \in [\ell-1]$.
We say that $\cP=p_1,\ldots,p_\ell$ is a $p_1$-$p_\ell$-walk and call $\ell-1$ the \emph{length} of $\cP$.
For walks $\cP = p_1,\ldots ,p_\ell$ and $\cQ=q_1,\ldots ,q_m$ such that $p_\ell = q_1$, we define $\cP \circ \cQ:= p_1,\ldots, p_\ell,q_2, \ldots, q_m$.
We say that two walks $\cP=p_1,\ldots, p_\ell$ and $\cQ=q_1,\ldots, q_m$ \emph{avoid each other} if $\ell=m$, $p_1 \neq q_1$, and $p_iq_{i+1},q_ip_{i+1} \not\in E(H)$ for every $i \in [\ell-1]$.
For two sets $A,B$ of vertices of a graph, we say that they are \emph{anticomplete}, if there is no edge with one endvertex in $A$ and another one~in~$B$.

We call the set of three vertices $T$ of a bipartite graph $H$ a \emph{special triple} if there exists an asteroid $(\{u_0,u_1,\ldots,u_{2k}\},\{v_0,v_1,\ldots,v_{2k}\})$ (we use the notation introduced in \cref{def:edge-ast}), such that $T=\{u_0,u_1,u_{k+1}\}$. Observe that the neighborhoods of vertices of every special triple are pairwise incomparable, as edges $u_0v_0$, $u_1v_1$, and $u_{k+1}v_{k+1}$ induce a matching. 

Let $(\{u_0,u_1,\ldots,u_{2k}\},\{v_0,v_1,\ldots,v_{2k}\})$ be an asteroid, and let $P_{i,i+1}$ for $i \in \{0,1,\ldots,2k\}$ denote the paths of this asteroid, satisfying \cref{def:edge-ast}.
Note that $(\{u_{0},u_{2k},\ldots,u_{1}\},\allowbreak \{v_0,v_{2k},\ldots,v_{1}\})$ is also an asteroid: we can use the same paths as in the first one, but in the reverse direction.
We will refer to this second asteroid as a \emph{reversed asteroid}.

\begin{proposition} \label{prop:special-triples}
Let $(\{u_0,u_1,\ldots,u_{2k}\},\{v_0,v_1,\ldots,v_{2k}\})$ be an asteroid in a bipartite graph $H$. Then each of the sets $\{u_0,u_1,u_{k+1}\},\{u_0,u_{2k},u_{k}\},\{v_0,v_1,v_{k+1}\}$ and $\{v_0,v_{2k},v_{k}\}$ is a special triple of $H$.
\end{proposition}
\begin{proof}
The fact that $\{u_0,u_1,u_{k+1}\}$ is a special triple follows from the definition. The set $\{u_0,u_{2k},u_{k}\}$ is a special triple because there exists a reversed asteroid $(\{u_0,u_{2k},\ldots,u_1\},\allowbreak\{v_0,v_{2k},\ldots,v_1\})$. 

Now let $P_{i,i+1}$ for some $i \in \{0,\ldots,2k\}$ be a path satisfying \cref{def:edge-ast} for the asteroid $(\{u_0,u_1,\ldots,u_{2k}\},\allowbreak\{v_0,v_1,\ldots,v_{2k}\})$. As $u_iv_i, u_{i+1}v_{i+1} \in E(H)$, it is straightforward to verify that there exists $v_i$-$v_{i+1}$-path $P'_{i,i+1}$, such that $\{u_{i+k+1},v_{i+k+1}\}$ and $\{u_{i},u_{i}\}\cup V(P'_{i,i+1})$ are anticomplete, and, if $i \not\in \{0,2k\}$, then $\{u_0,v_0\}$ and $\{u_{i},u_{i}\}\cup V(P'_{i,i+1})$ are also anticomplete. From this we can conclude that $(\{v_0,v_1,\ldots,v_{2k}\},\{u_0,u_1,\ldots,u_{2k}\})$ is also an asteroid, and $\{v_0,v_1,v_{k+1}\}$ is a special triple. The fact that $\{v_0,v_{2k},v_{k}\}$ is a special triple comes from combining the previous arguments. 
\end{proof}

To make the proof of \cref{lemma:extendedP4exists} easier, we first prove the following auxiliary lemma.

\begin{lemma} \label{lemma:ineq:five}
Let $H$ be a bipartite, connected graph which contains an asteroid $(U,V)$.
Then there exists a special triple $T$, an extended $P_4$ gadget~$(a,b,c,d,e)$ in~$H$ and:
\begin{enumerate}
\item walks $\cA,\cC,\cE$, starting, respectively, in $a$, $c$, and $e$, and terminating in distinct elements of $T$, such that each two of these walks avoid each other,
\item walks $\cB,\cD$, starting, respectively, in $b$ and $d$, and terminating in distinct elements of $T$, such that $\cB$ and $\cD$ avoid each other.
\end{enumerate} 
\end{lemma}
\begin{proof}
Recall that in the definition of an extended $P_4$ gadget~$(a,b,c,d,e)$ we require that the appropriate pairs of neighborhoods are incomparable. Actually, we will show a stronger property, i.e., that each of $a,c,e$ has a \emph{private neighbor}, which is non-adjacent to the other two vertices. We extend this notion and call vertices in $N_H(b) \setminus N_H(d)$ and in $N_H(d) \setminus N_H(b)$ \emph{private neighbors} of $b$ and $d$, respectively.

Define $F$ to be a minimal induced subgraph of $H$ which contains any asteroid, and let $(U,V)=(\{u_0,u_1,\ldots,u_{2k}\},\{v_0,v_1,\ldots,v_{2k}\})$ be an asteroid in $F$. {Notice that if the neighborhoods of some vertices are incomparable in $F$, then so are the neighborhoods of these vertices in $H$.} 

For every $i \in \{0,\ldots,2k\}$ we define $P_{i,i+1}$ as follows.
First, we choose $\widetilde{P}_{i,i+1}$ to be a shortest one from all $\{u_i,v_i\}$-$\{u_{i+1},v_{i+1}\}$-paths in $F$ that are anticomplete to $\{u_{i+k+1},v_{i+k+1}\}$ and, if $i \not\in\{0,2k\}$, also to $\{u_0,v_0\}$.
We know that at least one such path exists by the definition of an asteroid.
Clearly, exactly one of the vertices $u_i,v_i$ and exactly one of the vertices $u_{i+1},v_{i+1}$ belong to $\widetilde{P}_{i,i+1}$ (as endvertices).
Now if $u_i$ (respectively $u_{i+1}$) does not belong to $\widetilde{P}_{i,i+1}$, append it as the first (resp. last) vertex.
This way we obtain $P_{i,i+1}$. Observe that by the choice of $\widetilde{P}_{i,i+1}$ the path $P_{i,i+1}$ is induced.
 
The minimality of $F$ implies that every vertex of $F$ belongs to $(U \cup V)$ or at least one $P_{i,i+1}$. For every $i$ we define $\pstar_{i,i+1}:=V(P_{i,i+1}) \setminus \{u_i,v_i,u_{i+1},v_{i+1}\}$. Clearly this set induces a path in $F$.
Similarly, we define the set $\ovp_{i,i+1}:=V(P_{i,i+1}) \cup \{v_i,v_{i+1}\}$ and note that $\ovp_{0,1}$ and $\ovp_{2k,0}$ also induce paths in $F$.
Indeed, let us consider $\ovp_{0,1}$, the case of $\ovp_{2k,0}$ is symmetric. Recall that $P_{0,1}$ is an induced path. Furthermore, if $v_0$ (resp. $v_1$) does not belong to $P_{0,1}$, then it is non-adjacent to every vertex from $\pstar_{0,1}$ (by the minimality of $\widetilde{P}_{0,1}$) and also to $u_1$ (resp. $u_0$) by the property \eqref{it:special} in \cref{def:edge-ast}.
If it does not lead to confusion, we will sometimes identify sets $\pstar_{i,i+1}$, $\ovp_{2k,0}$ and $\ovp_{0,1}$ with the paths induced by these sets.



In the proof we will consider several cases.
First, suppose that $|\pstar_{0,1}| \geq 2$ or $|\pstar_{2k,0}| \geq 2$. Let us describe the first case, as the other one is symmetric -- we just need to consider the reversed asteroid.
Note that the first two vertices of $\ovp_{0,1}$ are either $u_0,v_0$ or $v_0,u_0$.
Consider the first case, as the other one is symmetric, with roles of $u$'s and $v$'s switched (recall that by~\cref{prop:special-triples} the set $\{v_0,v_1,v_{k+1}\}$ is also a special triple).

Since $|\pstar_{0,1}| \geq 2$, we know that there are vertices $b,c,d$, such that $\ovp_{0,1}$ starts with $u_0,v_0,b,c,d$, and $b,c \in \pstar_{0,1}$ and $d \in \pstar_{0,1} \cup \{u_1\}$.
We define an extended $P_4$ gadget to be the tuple $(v_0,b,c,d,v_{k+1})$ (recall that $\ovp_{0,1}$ is an induced path).
The private neighbors of $v_0,c,v_{k+1}$ are, respectively, $u_0,d,u_{k+1}$. The private neighbor of $b$ is $v_0$, and the private neighbor of $d$ is its successor on $\ovp_{0,1}$, i.e., the fourth vertex of $\pstar_{0,1}$, or $v_1$ if $|\pstar_{0,1}| < 4$.

Let $\cR$ be the shortest $d$-$v_1$-walk using consecutive vertices of $\ovp_{0,1}$. Note that its length is at least one. Define $\cD:=\cR \; \circ \; v_1,u_1$ and $\cB = b,v_0,u_0,\ldots,v_0,u_0$, so that $\cD$ and $\cB$ have equal lengths.
Similarly we define $\cA := v_0,u_0,\ldots,v_0,u_0$, $\cC := c,d \; \circ \; \cD \; \circ v_1,u_1$, and $\cE := v_{k+1},u_{k+1},\ldots,v_{k+1},u_{k+1}$, so that they have equal lengths. It is straightforward to verify that these walks satisfy the conditions in the lemma.

So we can assume that $\pstar_{0,1}$ has at most one vertex, and since $\{u_0,v_0\}$ must be anticomplete to $\{u_1,v_1\}$, we conclude that $\pstar_{0,1}$ contains exactly one vertex, say $x$.

Repeating the same argument for the reversed asteroid, we obtain that $\pstar_{2k,0}$ has exactly one vertex, say $x'$ (it might happen that $x=x'$).

Let us assume that $\ovp_{0,1}$ starts with $u_0,v_0$, as the other case is symmetric.
This means that the consecutive vertices of $\ovp_{0,1}$ are $u_0,v_0,x,v_1,u_1$.
Let $Q$ be a shortest $\{u_1,v_1\}$-$\{u_{k+1},v_{k+1}\}$-path contained in $V(F) \setminus \left( \{u_0,v_0\} \cup \pstar_{2k,0} \cup \pstar_{0,1} \right) = V(F) \setminus N_F[\{u_0,v_0\}]$, it exists by the definition of an asteroid.
Note that $Q$ is induced and anticomplete to $\{u_0,v_0\}$, and exactly one of the vertices $u_1,v_1$ and exactly one of the vertices $u_{k+1},v_{k+1}$ belong to $Q$. 
Let $Q^*:=V(Q) \setminus \{u_1,v_1,u_{k+1},v_{k+1}\}$ and note that $Q^*$ must be non-empty, because $\{u_1,v_1\}$ is anticomplete to $\{u_{k+1},v_{k+1}\}$. Again, we will identify $Q^*$ with the subpath of $Q$ induced by $Q^*$.
For a vertex $v \in Q$, let $\cQ_v$ be the shortest $v$-$u_{k+1}$-walk, which uses only vertices of $Q \cup \{u_{k+1}\}$.

If $x$ has no neighbors in $Q^*$, we can define our gadget to be the tuple $(v_0,x,v_1,d,v_{k+1})$, where $d$ is the first vertex of $Q^*$ if $v_1 \in Q$ or $d=u_1$ if $u_1 \in Q$. As $\{v_0,v_1,v_{k+1}\}$ is a special triple, each of these vertices has a private neighbor. The private neighbor of $x$ is $v_0$, and the private neighbor of $d$ is its successor on $Q$.
We define walks $\cA:=v_0,u_0$ and $\cC:=v_1,u_1$ and $ \cE:=v_{k+1},u_{k+1}$.
Moreover, we define $\cD:=\cQ_d$ and $\cB:=x,v_0,u_0,\ldots, u_0$, so that they are of equal length.

So we can assume that $x$ has a neighbor in $Q^*$.
Denote by $y$ the last neighbor of $x$ in $Q^*$ and by $q$ the successor of $y$ on $Q$; it exists, because $Q$ terminates at one of $u_{k+1},v_{k+1}$ and $y \neq v_{k+1}$. Clearly, $q \not\in N_F(v_0)$, because it belongs to $Q$ which is anticomplete to $\{u_0,v_0\}$; also $q \not\in N_F(v_1)$ by the choice of $Q$: otherwise we would have chosen a shorter path starting with $v_1$ and then using the consecutive elements of $\cQ_q$.
Note that $q$ might be equal to $u_{k+1}$.
We now branch on two cases, depending on the size of $Q^*$.

\textbf{Case 1: $|Q^*| \geq 2$.} If $y \not\in N_F(u_1)$, we define our gadget to be the tuple $(v_0,x,y,q,v_1)$.
The private neighbors of $v_0,v_1$, and $y$ are $u_0,u_1,$ and $q$, respectively.
The private neighbor of $x$ is $v_0$.
The private neighbor of $q$ is its successor on $Q$ (if $q \in Q^*$), or $v_{k+1}$ if $q=u_{k+1}$.

We define walks $\cA := v_0,u_0,\ldots,u_0$ and $\cC := \cQ_y$ and $\cE := v_1,u_1,\ldots,u_1$, so that they have equal lengths. Similarly, we define $\cB := x,v_0,u_0,\ldots,v_0,u_0$ and $\cD = \cQ_q \circ u_{k+1},v_{k+1},u_{k+1}$, so that they have equal lengths. We appended $u_{k+1},v_{k+1},u_{k+1}$ at the end of $\cD$, so that we do not need to treat the case that $q=u_{k+1}$ separately.

So assume that $y \in N_F(u_1)$, so it is the first vertex of $Q^*$. Note that in this case $q \in Q^*$, so, in particular, $q \notin \{u_{k+1},v_{k+1}\}$. Therefore $q$ has a successor $q'$ in $Q$. 
Then we take the tuple $(v_0,x,y,q,v_{k+1})$ with corresponding private neighbors $u_0,v_0,u_1,q'$, and $u_{k+1}$.
We define walks as follows: $\cA:=v_0,u_0$ and $\cC := y,u_1$, and $\cE:=v_{k+1},u_{k+1}$. Furthermore, we define $\cB:=x,v_0,u_0,\ldots,u_0$ and $\cD := \cQ_q$. Note that in the currently considered case the length of $\cD$ is at least two.

\textbf{Case 2: $|Q^*| = 1$.} So we are left with the case that $Q^*$ consists only of a vertex $y$, which is adjacent to both $x$ and $u_{k+1}$. 
Recall that $Q \subseteq \ovp_{1,2} \cup \ldots \cup \ovp_{k,k+1} \cup \ovp_{k+1,k+2} \ldots \ovp_{2k-1,2k}$, which means that there is non-empty  $I \subseteq [2k-1]$, such that $y \in \bigcap_{i \in I}\ovp_{i,i+1}$.

First suppose that there is some $i \in I \setminus \{k\}$. This means that there exists $\ell=i+k+1 \neq 0$ such that $\{u_{\ell},v_{\ell},u_0,v_0\}$ is anticomplete to $\ovp_{i,i+1}$.
%
Furthermore, since $\ovp_{i,i+1}$ is connected and contains at least 2 vertices, 
$y \in \ovp_{i,i+1}$ has a neighbor $r$ in $\ovp_{i,i+1}$.
We define the extended $P_4$ gadget as the tuple $(v_0,x,y,u_{k+1},v_{\ell})$.
The corresponding private neighbors are $u_0,v_0,r,v_{k+1}$, and $u_{\ell}$. 
{Recall that $x \in P_{0,1}$, so it must be non-adjacent to $v_{k+1}$.} 

We define walks $\cB:=x,v_0,u_0$ and $\cD:=u_{k+1},v_{k+1},u_{k+1}$.
The definition of the remaining three walks is more intricate. 
Let $\cR$ be the shortest $y$-$u_i$-walk using consecutive vertices of $P_{i-1,i}$
For $j \in [2k]$, let $\cP_{j,j+1}$ be the shortest $u_{j}$-$u_{j+1}$-walk using vertices of $P_{j,j+1}$. By $\cP_{j+1,j}$ we denote the walk $\cP_{j,j+1}$ in the reversed order.

If $i \in [k-1]$, we set:
\begin{equation*}
\begin{alignedat}{7}
\cA :=&  v_0,u_0,v_0,..,u_0 &&\circ u_0,v_0,..,u_0 &&\circ u_0,v_0,..,u_0 &&\circ \ldots &&\circ u_0,v_0,..,u_0 &&\circ u_0,v_0,..,u_0 &&\circ u_0,v_0,..,u_0,\\
\cC :=& \cR &&\circ u_i,v_i,..,u_i &&\circ \cP_{i,i-1} &&\circ \ldots &&\circ u_2,v_2,..,u_2  &&\circ \cP_{2,1} &&\circ u_1,v_1,..,u_1, \\
\cE :=& v_{\ell},u_{\ell},v_{\ell},..,u_{\ell} &&\circ \cP_{\ell,\ell-1} &&\circ u_{\ell-1},v_{\ell-1},..,u_{\ell-1} &&\circ \ldots &&\circ \cP_{k+3,k+2} &&\circ u_{k+2},v_{k+2},..,u_{k+2} &&\circ \cP_{k+2,k+1},
\end{alignedat}
\end{equation*} 
where the lengths of particular segments are adjusted, so that the subwalks in the same columns have the same length.
By the definition of an edge asteroid, for each $j$ the set $\{u_j,v_j\}$ is anticomplete to $\cP_{j+k,j+k+1}$.
Furthermore, as $\cC$ and $\cE$ do not use the vertices from $N_F[\{u_0,v_0\}]$, each two of $\cA,\cC$, and $\cE$ avoid each other.

If $i \in [2k-1] \setminus [k]$, we set:
\begin{equation*}
\begin{alignedat}{6}
\cA :=&  v_0,u_0,v_0,..,u_0 &&\circ u_0,v_0,..,u_0 &&\circ u_0,v_0,..,u_0 &&\circ \ldots &&\circ u_0,v_0,..,u_0 &&\circ u_0,v_0,..,u_0,\\
\cC :=& \cR &&\circ u_{i},v_{i},..,u_{i} &&\circ \cP_{i,i-1} &&\circ \ldots &&\circ u_{k+2},v_{k+2},..,u_{k+2}  &&\circ \cP_{k+2,k+1}, \\
\cE :=& v_{\ell},u_{\ell},v_{\ell},..,u_{\ell} &&\circ \cP_{\ell,\ell-1} &&\circ u_{\ell-1},v_{\ell-1},..,u_{\ell-1} &&\circ \ldots &&\circ \cP_{2,1} &&\circ u_{1},v_{1},..,u_{1}.
\end{alignedat}
\end{equation*} 
The argument that these walks avoid each other is analogous to the previous case.

\medskip

So finally we assume that $I = \{k\}$, i.e., $y \in \ovp_{k,k+1}$ and $y \not\in \ovp_{j,j+1}$ for every $j \neq k$.
Note that this in particular means that $y \neq v_j$ for any $j \in \{0,\ldots,2k\}$, as each $v_j$ belongs to $\ovp_{j-1,j} \cup \ovp_{j,j+1}$. So $y \in \pstar_{k,k+1}$.

Let us define $K:=\ovp_{1,2} \cup \ldots \cup \ovp_{k-1,k}$ and $K':=\ovp_{k+1,k+2} \cup \ldots \cup \ovp_{2k-1,2k}$.
Note that each of them induces a connected subgraph of $F$ and each of them is anticomplete to $\{u_0,v_0\}$.
We claim that $K$ is anticomplete to $K'$. Indeed, note that otherwise there is a $u_k$-$u_{k+1}$ path $P'_{k,k+1}$ in $F[K \cup K']$, which does not use $y$ and is anticomplete to $\{u_0,v_0\}$. This implies $(U,V)$ is an asteroid in $F - y$, where the path between $u_k$ and $u_{k+1}$ is $P'_{k,k+1}$. This contradicts the definition of $F$.

\medskip

Recall that our argument can be repeated for the reversed asteroid, as we did when defining $x'$.
So let $y'$ be an analogue of $y$, i.e., the last neighbor of $x'$ on the shortest $\{u_{2k},v_{2k}\}$-$\{u_k,v_k\}$-path contained in $V(F) \setminus N_F[\{u_0,v_0\}]$.
Observe that for the reversed asteroid the bipartition classes might be switched, i.e., it is possible that $x'$ is adjacent to $u_0$ and $u_{2k}$, and $y'$ is adjacent to $v_{2k}$ and $v_k$. Furthermore, the sets $\{x,y\}$ and $\{x',y'\}$ might overlap. However, we know that $y,y' \in \pstar_{k,k+1}$.

If $|\pstar_{k,k+1}|\geq 2$, then $\ovp_{k,k+1}$ must be an induced path on at least 6 vertices.
If $v_{i+1} \in P_{k,k+1}$, we denote by  $z,a,b,c,v_{k+1}$ the last five consecutive vertices of $P_{k,k+1}$, and define the gadget as $(a,b,c,v_{k+1},u_0)$. The corresponding private neighbors are $z,a,v_{k+1},u_{k+1}$, and $v_0$.

Denote by $\cR$ the shortest $a$-$u_k$-walk using only the vertices of $\ovp_{k,k+1}$.
Let $\cK$ be any $u_k$-$u_{1}$-walk contained in $K$, recall that $\cK$ is anticomplete to $\{u_{k+1},v_{k+1},u_0,v_0\}$.
Then we define, respectively, $\cA:=\cR \circ \cK$ and $\cC:=c,v_{k+1},u_{k+1},\ldots,u_{k+1}$, and $\cE:= u_0, v_0,\ldots, u_0$, so that they have equal lengths.
Similarly, we define $\cB:=b,a \circ \cR \circ \cK$ and $\cD:=v_{k+1},u_{k+1},\ldots, u_{k+1}$, so that they have equal lengths.

If $v_{i+1} \not\in P_{k,k+1}$, the gadget is $(a,b,c,u_{k+1},v_0)$, where $a,b,c,u_{k+1}$ are last four vertices of $P_{k,k+1}$. The remaining argument is analogous.

\medskip

So we are left with the case $|\pstar_{k,k+1}|=1$, and since $y, y' \in \pstar_{k,k+1}$, we must have that $y=y'$.
This also implies that $x'$ is adjacent to $v_0$ and $v_{2k}$.

Observe that if $x=x'$, then $x$ is also non-adjacent to $v_k$ and we define the extended $P_4$ to be $(u_{k+1},y,x,v_0,u_k)$, where the corresponding private neighbors are $v_{k+1},u_{k+1},v_0,u_0,v_k$.
Define walks $\cB:=y,u_{k+1},v_{k+1},\ldots,u_{k+1}$, and $ \cD:=v_0,u_0,\ldots,u_0$ of equal length.
For $\cK$ being any $u_k$-$u_{1}$-walk contained in $K$, we define $\cE:=\cK$ and walks $\cA:=u_{k+1},v_{k+1},\ldots,u_{k+1}$ and $\cC:=x,v_0,u_0,\ldots,u_0$ of same length as $\cE$.

This means that we can assume that $x \neq x'$, and thus $F$ contains the structure depicted in~\cref{fig:final-asteroid}~(left).

\begin{figure}[ht]
\begin{center}
\begin{tikzpicture}[every node/.style={draw,circle,fill=white,inner sep=0pt,minimum size=5pt},every loop/.style={},scale=0.6]
\node[label=above:$u_0$] (u0) at (0,4) {};
\node[fill=black,label=above:$v_0$] (v0) at (1,4) {};
\node[label=above:$x$] (x) at (2,4) {};
\node[fill=black,label=above:$v_1$] (v1) at (3,4) {};
\node[label=above:$u_1$] (u1) at (4,4) {};

\node[label=left:$x'$] (x') at (1,3) {};
\node[fill=black,label=left:$v_{2k}$] (v2k) at (1,2) {};
\node[label=left:$u_{2k}$] (u2k) at (1,1) {};

\node[fill=black,label=285:$y$] (y) at (2,3) {};

\node[label=below:$u_k$] (uk) at (3,3) {};
\node[fill=black,label=below:$v_k$] (vk) at (4,3) {};

\node[label=right:$u_{k+1}$] (uk1) at (2,2) {};
\node[fill=black,label=right:$v_{k+1}$] (vk1) at (2,1) {};

\draw[line width=1.5] (u2k) -- (y) -- (u1);
\draw[line width=1.5] (u0)--(v0)--(x)--(v1)--(u1); 
\draw[line width=1.5] (v0)--(x')--(v2k)--(u2k);

\draw[line width=1.5] (x')--(y)--(uk)--(vk);
\draw[line width=1.5] (x)--(y)--(uk1)--(vk1);

\draw[dashed] (x) -- (vk);
\draw[dashed] (x) -- (v2k);
\draw[dashed] (x') -- (v1);
\draw[dashed] (x') -- (vk1);
\end{tikzpicture}
\hskip 3cm
\begin{tikzpicture}[every node/.style={draw,circle,fill=white,inner sep=0pt,minimum size=5pt},every loop/.style={},scale=0.6]
\node[label=above:$u_0$] (u0) at (0,4) {};
\node[fill=black,label=above:$v_0$] (v0) at (1,4) {};
\node[label=above:$x$] (x) at (2,4) {};

\node[label=left:$x'$] (x') at (1,3) {};

\node[fill=black,label=285:$y$] (y) at (2,3) {};

\node[label=below:$u_k$] (uk) at (3,3) {};
\node[fill=black,label=below:$v_k$] (vk) at (4,3) {};

\node[label=right:$u_{k+1}$] (uk1) at (2,2) {};
\node[fill=black,label=right:$v_{k+1}$] (vk1) at (2,1) {};

\draw[line width=1.5] (u0)--(v0)--(x);
\draw[line width=1.5] (v0)--(x');

\draw[line width=1.5] (x')--(y)--(uk)--(vk);
\draw[line width=1.5] (x)--(y)--(uk1)--(vk1);

\draw[line width=1.5]  (x) -- (vk);
\draw[line width=1.5]  (x') -- (vk1);
\end{tikzpicture}
\end{center}
\caption{Left: The structure in the last case in the proof of \cref{lemma:ineq:five}. Dashed edges might exist, but do not have to. Right: A smaller asteroid exists if $xv_k \in E(H)$ and $x'v_{k+1} \in E(H)$.}\label{fig:final-asteroid}
\end{figure}
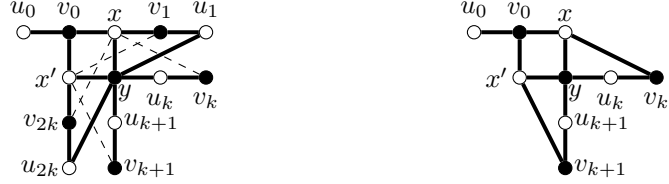

Now observe that if $xv_k \in E(H)$ and $x'v_{k+1} \in E(H)$, then the proper subgraph of $F$ induced by $\{u_0,v_0,x,v_k,u_k,y,u_{k+1},v_{k+1},x'\}$ contains an asteroid $(\{u_0,u_k,u_{k+1}\},\{v_0,v_k,v_{k+1}\})$, which contradicts the minimality of $F$ (see \cref{fig:final-asteroid}~(right)).

So suppose that at least one of these edges, say $x'v_{k+1}$, does not exist (the other case is symmetric).
The minimality of $F$ implies that the edge $x'v_1$ also does not exist: otherwise $F-x$ still contains the asteroid $(U,V)$, where the path between $u_0$ and $v_i$ is $u_0,x',v_1$.
In such a case we take the tuple $(u_{k+1},y,x',v_0,u_1)$, where their corresponding private neighbors are $v_{k+1},u_{k+1},v_0,u_0$, and $v_1$.
The walks are $\cA:=u_{k+1},v_{k+1},u_{k+1}$, $\cB:=y, u_{k+1}$, $\cC:=x',v_0,u_0$, $\cD:=v_0,u_0$, and $\cE:=u_1,v_1,u_1$.
This completes the proof of the lemma.
\end{proof}

Now we proceed to the proof of \cref{lemma:extendedP4exists}.
\begin{proof}[Proof of \cref{lemma:extendedP4exists}]
If $H$ contains an induced cycle with consecutive vertices  $x_0,x_1,\ldots,x_{k-1},x_0$ for $k \in \{6,8\}$, we define an extended $P_4$ gadget to be the tuple $(x_0,x_1,x_2,x_3,x_4)$. Clearly, for every pair of distinct $i,j \in \{0,\ldots,k\}$ we have $N_H(x_i) \neq N_H(x_j)$, as they belong to an induced cycle of length more than 4. Then the appropriate instances $(G_{\{x_0,x_2,x_4\}},L)$ are shown in \cref{fig:cycles-gadgets}.

\begin{figure}[h]
\begin{center}
\begin{tikzpicture}[every node/.style={draw,circle,fill=white,inner sep=0pt,minimum size=6pt}]
\draw plot [smooth] coordinates {(0,1) (0.25,.5) (.5,0.25) (1,0)};
\draw plot [smooth] coordinates {(0,1) (0.25,1.5) (.5,1.75) (1,2)};
\draw plot [smooth] coordinates {(3,0) (3.5,0.25) (3.75,.5) (4,1)};
\draw plot [smooth] coordinates {(3,2) (3.5,1.75) (3.75,1.5) (4,1)};
\node[fill=gray] (x) at (0,1) {};
\node[label={[label distance=-3mm]90:\colorbox{gray!30}{\scriptsize{$x_3,x_5$}}}] (a1) at (1,0) {};
\node[label={[label distance=-3mm]90:\colorbox{gray!30}{\scriptsize{$x_0,x_4$}}}] (a2) at (2,0) {};
\node[label={[label distance=-3mm]90:\colorbox{gray!30}{\scriptsize{$x_1,x_5$}}}] (a3) at (3,0) {};
\node[label={[label distance=-3mm]90:\colorbox{gray!30}{\scriptsize{$x_1,x_5$}}}] (b1) at (1,1) {};
\node[label={[label distance=-3mm]90:\colorbox{gray!30}{\scriptsize{$x_0,x_2$}}}] (b2) at (2,1) {};
\node[label={[label distance=-3mm]90:\colorbox{gray!30}{\scriptsize{$x_1,x_3$}}}] (b3) at (3,1) {};
\node[label={[label distance=-3mm]90:\colorbox{gray!30}{\scriptsize{$x_1,x_3$}}}] (c1) at (1,2) {};
\node[label={[label distance=-3mm]90:\colorbox{gray!30}{\scriptsize{$x_2,x_4$}}}] (c2) at (2,2) {};
\node[label={[label distance=-3mm]90:\colorbox{gray!30}{\scriptsize{$x_3,x_5$}}}] (c3) at (3,2) {};
\node[fill=gray] (y) at (4,1) {};
\draw (a1) -- (a2) -- (a3);
\draw (x) -- (b1) -- (b2) -- (b3) -- (y);
\draw (c1) -- (c2) -- (c3);

\end{tikzpicture}
\hskip 1cm
\begin{tikzpicture}[every node/.style={draw,circle,fill=white,inner sep=0pt,minimum size=6pt},scale=1]
\draw plot [smooth] coordinates {(0,1) (0.25,.5) (.5,0.25) (1,0)};
\draw plot [smooth] coordinates {(0,1) (0.25,1.5) (.5,1.75) (1,2)};
\draw plot [smooth] coordinates {(5,0) (5.5,0.25) (5.75,.5) (6,1)};
\draw plot [smooth] coordinates {(5,2) (5.5,1.75) (5.75,1.5) (6,1)};
\node[fill=gray] (x) at (0,1) {};
\node[label={[label distance=-3mm]90:\colorbox{gray!30}{\scriptsize{$x_3,x_7$}}}] (a1) at (1,0) {};
\node[label={[label distance=-3mm]90:\colorbox{gray!30}{\scriptsize{$x_4,x_6$}}}] (a2) at (2,0) {};
\node[label={[label distance=-3mm]90:\colorbox{gray!30}{\scriptsize{$x_3,x_5$}}}] (a3) at (3,0) {};
\node[label={[label distance=-3mm]90:\colorbox{gray!30}{\scriptsize{$x_2,x_4$}}}] (a4) at (4,0) {};
\node[label={[label distance=-3mm]90:\colorbox{gray!30}{\scriptsize{$x_1,x_3$}}}] (a5) at (5,0) {};
\node[label={[label distance=-3mm]90:\colorbox{gray!30}{\scriptsize{$x_1,x_5$}}}] (b1) at (1,1) {};
\node[label={[label distance=-3mm]90:\colorbox{gray!30}{\scriptsize{$x_2,x_6$}}}] (b2) at (2,1) {};
\node[label={[label distance=-3mm]90:\colorbox{gray!30}{\scriptsize{$x_1,x_7$}}}] (b3) at (3,1) {};
\node[label={[label distance=-3mm]90:\colorbox{gray!30}{\scriptsize{$x_0,x_2$}}}] (b4) at (4,1) {};
\node[label={[label distance=-3mm]90:\colorbox{gray!30}{\scriptsize{$x_1,x_3$}}}] (b5) at (5,1) {};
\node[label={[label distance=-3mm]100:\colorbox{gray!30}{\scriptsize{$x_1,x_5,x_7$}}}] (c1) at (1,2) {};
\node[label={[label distance=-3mm]90:\colorbox{gray!30}{\scriptsize{$x_0,x_6$}}}] (c2) at (2,2) {};
\node[label={[label distance=-3mm]90:\colorbox{gray!30}{\scriptsize{$x_5,x_7$}}}] (c3) at (3,2) {};
\node[label={[label distance=-3mm]90:\colorbox{gray!30}{\scriptsize{$x_4,x_6$}}}] (c4) at (4,2) {};
\node[label={[label distance=-3mm]80:\colorbox{gray!30}{\scriptsize{$x_3,x_5,x_7$}}}] (c5) at (5,2) {};
\node[fill=gray] (y) at (6,1) {};
\draw (a1) -- (a2) -- (a3) -- (a4) -- (a5);
\draw (x) --  (b1) -- (b2) -- (b3) -- (b4) -- (b5)  -- (y);
\draw (c1) -- (c2) -- (c3) -- (c4) -- (c5);
\end{tikzpicture}
\end{center}
\caption{A \listHcoloring instance~$(G_{\{x_0,x_2,x_4\}}, L)$ satisfying the statement 1. of \cref{lemma:extendedP4exists} in case that $H$ contains an induced $C_6$ (left) or an induced $C_8$ (right).
Vertices $\gamma_1, \gamma_2$ are marked gray, and $L(\gamma_1)=L(\gamma_2)=\{x_0,x_2,x_4\}$.}\label{fig:cycles-gadgets}
\end{figure}
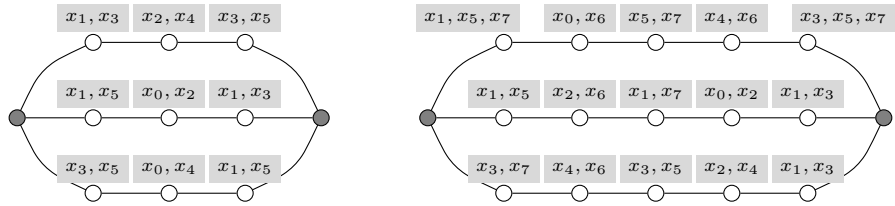

As the cycles are symmetric, we can obtain the instance $(G_{\{x_1,x_3\}},L)$ for $C_6$ and $C_8$ by taking the same graph as $G_{\{x_0,x_2,x_4\}}$, removing $x_4$ from the lists of $\gamma_1,\gamma_2$, and replacing $x_i$ by $x_{i+1}$ (modulo $k$) for every element of every list.

Observe that every induced cycle in $H$ on at least 10 vertices, with consecutive vertices $x_0,x_1,\ldots,x_{k-1},x_0$, contains an asteroid $(\{x_0,x_4,x_6\},\{x_1,x_3,x_7\})$: the paths $P_{0,1}:=x_0,x_1,x_2,x_3,x_4$, and $P_{1,2}:=x_4,x_5,x_6$, and $P_{2,0}:=x_6,x_7,\ldots,x_{k-1},x_0$ satisfy \cref{def:edge-ast}.
So now it is sufficient for consider the case that $H$ contains an asteroid. By \cref{lemma:ineq:five} we know that in such a case there exist:
\begin{enumerate}[(a)]
\item a special triple $T$,
\item an extended $P_4$ gadget $(a,b,c,d,e)$,
\item an injective function $\sigma : \{a,c,e\} \to T$,
\item an injective function $\pi : \{b,d\} \to T$,
\item walks $\cA,\cC$, and $\cE$, starting, respectively, in $a,c,$, and $e$, and terminating, respectively, in $\sigma(a),\sigma(c),\sigma(e)$, such that each two of $\cA,\cC$, and $\cE$ avoid each other,
\item walks $\cB,\cD$, starting, respectively, in $b$ and $d$, and terminating, respectively, in $\pi(b),\pi(d)$, such that $\cB$ and $\cD$ avoid each other.
\end{enumerate}
Let us show how to construct $(G_{\{a,c,e\}},L)$, the construction of $(G_{\{b,d\}},L)$ is analogous, we just need to use walks $\cB$ and $\cD$ instead of $\cA,\cC,\cE$.

Recall that $\cA,\cC$, and $\cE$ are of equal length, say $\ell$, i.e., each of them has $\ell+1$ vertices.
We define the instance $C(\cA,\cC,\cE):=(G,L)$ of \listHcoloring, such that $G$ is a path with consecutive vertices $y_1,y_2,\ldots,y_{\ell +1}$, and the list $L(y_i)$ contains the $i$-th vertex of $\cA$, the $i$-th vertex of $\cC$, and the $i$-th vertex of $\cE$.
Note that since walks $\cA,\cC,\cE$ avoid each other, for every $i \in [\ell+1]$ we have $|L(y_i)|=3$, and, in particular, $L(y_1)=\{a,c,e\}$ and $L(y_{\ell+1})=T$.

Furthermore, each list homomorphism $h$ from $C(\cA,\cC,\cE)$ to $H$ coincides either with one of $\cA,\cC,\cE$. More formally, we have the following:
\begin{enumerate}
\item for every $x \in \{a,c,e\}$, there is a list homomorphism $h_x \colon C(\{\cA,\cC,\cE\}) \to H$, such that $h_x(y_1)=x$ and $h_x(y_{\ell+1})=\sigma(x)$,~and
\item for any list homomorphism $h:C(\{\cA,\cC,\cE\}) \to H$ there is $x \in \{a,c,e\}$, such that $h(y_1)=x$ and $h(y_{\ell+1})=\sigma(x)$.
\end{enumerate}

Recall that $T = \{u_0,u_1,u_{k+1}\}$ for some asteroid $(\{u_0,u_1,\ldots,u_{2k}\},(\{v_0,v_1,\ldots,v_{2k}\})$ in $H$.
We need one more tool from the construction of Feder, Hell, and Huang \cite[Fig. 3]{Feder99}, which we call an \emph{unequal gadget}. The unequal gadget is an instance $(F, L)$ of \listHcoloring with two distinguished vertices~$\delta_1, \delta_2$, such $L(\delta_1)=L(\delta_2)=T$ and any function $f \colon \{\delta_1, \delta_2\} \to T$ can be extended to a proper list $H$-coloring of~$(F,L)$ if and only if~$f(\delta_1) \neq f(\delta_2)$.

Now to create an instance $(G_{\{a,c,e\}},L)$ satisfying the statement 1. of our lemma, we first introduce a copy $F$ of the unequal gadget with distinguished vertices $\delta_1$ and $\delta_2$.
Then we introduce two copies $C^{(1)},C^{(2)}$ of $C(\{\cA,\cC,\cE\})$.
For $i \in \{1,2\}$ we denote by $y^{(i)}_1,y^{(i)}_{\ell+1}$ the endvertices of $C^{(i)}$.
Then we identify vertices $y^{(i)}_{\ell+1}$ and $\delta_i$, as $L(\delta_i):=L(y^{(i)}_{\ell+1})=T$, and put $\gamma_i:=y^{(i)}_1$ (see \cref{fig:gadget-schem}). This completes the construction of $(G_{\{a,c,e\}},L)$. Clearly, $L(\gamma_1)=L(\gamma_2)=\{a,c,e\}$. 

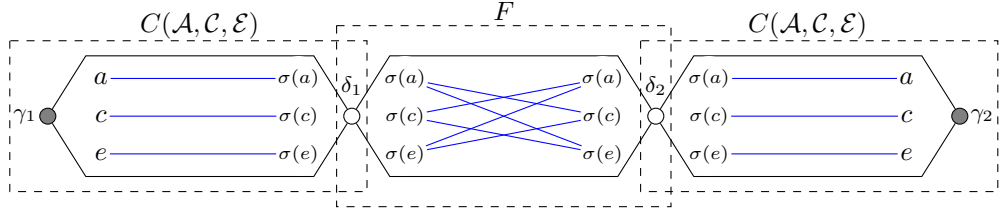
\begin{figure}[bt]
\begin{center}
\begin{tikzpicture}[every node/.style={draw,circle,fill=white,inner sep=0pt,minimum size=6pt}]
\node[draw=none] at (2,2.2) {$C(\cA,\cC,\cE)$};
\node[draw=none] at (10,2.2) {$C(\cA,\cC,\cE)$};
\node[draw=none] at (6,2.4) {$F$};
\draw[dashed]  (-.5,2) rectangle (4.2,0);
\draw[dashed]  (3.8,2.2) rectangle (8.2,-.2);
\draw[dashed]  (7.8,2) rectangle (12.5,0);
\node[label=180:\footnotesize{$\gamma_1$},fill=gray] (gamma1) at (0,1) {};
\node[label={[label distance=.1cm]90:\footnotesize{$\delta_1$}}] (delta1) at (4,1) {};
\node[label={[label distance=.1cm]90:\footnotesize{$\delta_2$}}] (gamma2) at (8,1) {};
\node[label=0:\footnotesize{$\gamma_2$},fill=gray] (delta2) at (12,1) {};


\draw (gamma1) -- (.5,1.8) -- (3.5,1.8) -- (delta1) -- (4.5,1.8) -- (7.5,1.8) -- (gamma2) -- (8.5,1.8) -- (11.5,1.8) -- (delta2);
\draw (gamma1) -- (.5,.2) -- (3.5,.2) -- (delta1) -- (4.5,.2) -- (7.5,.2) -- (gamma2) -- (8.5,.2) -- (11.5,.2) -- (delta2);
\node[draw=none] (c1) at (.7,.5) {$e$};
\node[draw=none] (b1) at (.7,1) {$c$};
\node[draw=none] (a1) at (.7,1.5) {$a$};

\node[draw=none] (tc1) at (3.3,.5) {\scriptsize{$\sigma(e)$}};
\node[draw=none] (tb1) at (3.3,1) {\scriptsize{$\sigma(c)$}};
\node[draw=none] (ta1) at (3.3,1.5) {\scriptsize{$\sigma(a)$}};

\node[draw=none] (tc2) at (4.7,.5) {\scriptsize{$\sigma(e)$}};
\node[draw=none] (tb2) at (4.7,1) {\scriptsize{$\sigma(c)$}};
\node[draw=none] (ta2) at (4.7,1.5) {\scriptsize{$\sigma(a)$}};

\node[draw=none] (tc3) at (7.3,.5) {\scriptsize{$\sigma(e)$}};
\node[draw=none] (tb3) at (7.3,1) {\scriptsize{$\sigma(c)$}};
\node[draw=none] (ta3) at (7.3,1.5) {\scriptsize{$\sigma(a)$}};

\node[draw=none] (tc4) at (8.7,.5) {\scriptsize{$\sigma(e)$}};
\node[draw=none] (tb4) at (8.7,1) {\scriptsize{$\sigma(c)$}};
\node[draw=none] (ta4) at (8.7,1.5) {\scriptsize{$\sigma(a)$}};

\node[draw=none] (c2) at (11.3,.5) {$e$};
\node[draw=none] (b2) at (11.3,1) {$c$};
\node[draw=none] (a2) at (11.3,1.5) {$a$};
\draw[color=blue] (a1) -- (ta1);
\draw[color=blue] (b1) -- (tb1);
\draw[color=blue] (c1) -- (tc1);


\draw[color=blue] (ta2) -- (tb3);
\draw[color=blue] (tb2) -- (tc3);
\draw[color=blue] (tc2) -- (ta3);

\draw[color=blue] (ta2) -- (tc3);
\draw[color=blue] (tb2) -- (ta3);

\draw[color=blue] (tc2) -- (tb3);

\draw[color=blue] (a2) -- (ta4);
\draw[color=blue] (b2) -- (tb4);
\draw[color=blue] (c2) -- (tc4);
\end{tikzpicture}
\end{center}
\caption{The construction of $(G_{\{a,c,e\}},L)$ as a composition of two copies of $C(\cA,\cC,\cE)$ and a copy of $F$. We have $L(\gamma_1)=L(\gamma_2)=\{a,c,e\}$ and $L(\delta_1)=L(\delta_2)=T=\{\sigma(a),\sigma(b),\sigma(c)\}$. Blue lines denote which mappings of $\gamma_1,\delta_1,\delta_2,\gamma_2$ to the vertices on their lists can be extended to a list homomorphism of particular gadgets.}
\label{fig:gadget-schem}
\end{figure}

To see that $(G_{\{a,c,e\}},L)$ satisfies the statement 1., assume that we have a homomorphism $h \colon (G_{\{a,c,e\}},L) \to H$, such that $h(\gamma_1)=h(\gamma_2)$. Then, by the properties of $C(\{\cA,\cC,\cE\})$, we must have $h(\delta_1)=h(\delta_2)$. But $\delta_1,\delta_2$ are distinguished vertices of an unequal gadget, so we have a contradiction. On the other hand, if we take some mapping $h' \colon \{\gamma_1,\gamma_2\} \to \{a,c,e\}$ with~$h'(\gamma_1) \neq h'(\gamma_2)$, we can always find $h^1$ and $h^2$, which are list homomorphisms from $C^{(1)}$ and $C^{(2)}$ to $H$, with the property that $h^1(\delta_1)=\sigma(h'(\gamma_1)) \neq \sigma(h'(\gamma_2)) =h^2(\delta_2)$.
Since $\sigma(h'(\gamma_1)) \neq \sigma(h'(\gamma_2))$, we can find a list homomorphism $f$ from $F$ to $H$, such that $f(\delta_1)=\sigma(h'(\gamma_1))$ and $f(\delta_2)=\sigma(h'(\gamma_2))$. Since homomorphisms $h^1,h^2$, and $f$ agree on common vertices, we can define $h$ to be the union of these three mappings.
This completes the proof.
\end{proof}

From the gadgets of Lemma~\ref{lemma:extendedP4exists}, we can also make efficient larger gadgets to enforce that in a large group of vertices, at least one vertex is not colored~$c$. The construction is an adaptation of a gadget due to Jaffke and Jansen~\cite{JaffkeJ17}.

\begin{lemma}\label{lemma:blocking:gadget}
Let~$H$ be a bipartite graph which contains an induced cycle of at least 6 vertices or an asteroid, and let~$(a,b,c,d,e)$ be an extended $P_4$ gadget in~$H$ as guaranteed by Lemma~\ref{lemma:extendedP4exists}. For any~$k \geq 2$ one can construct a consistent \listHcoloring instance~$(G, L)$ in polynomial time  containing~$k$ distinguished vertices~$\gamma_1, \ldots, \gamma_k$ such that~$|V(G)| \in \Oh(k)$, and such that a mapping~$f \colon \{\gamma_1, \ldots, \gamma_k\} \to \{a,c,e\}$ can be extended to a proper list $H$-coloring of~$(G,L)$ if and only if there exists an~$i \in [k]$ with~$f(\gamma_i) \neq c$.
\end{lemma}

\begin{proof}
The construction is a small adaptation of a gadget due to Jaffke and Jansen~\cite{JaffkeJ17}, which we present here for completeness.

Let~$T$ be the complete graph (triangle) on vertex set~$\{1,2,3\}$. We first show how to construct an instance~$(G', L')$ of \listcoloring{$T$} with~$k$ distinguished vertices~$\gamma_1, \ldots, \gamma_k$, such that a mapping~$f \colon \{\gamma_1, \ldots, \gamma_k\} \to \{1,2,3\}$ can be extended to a proper \listcoloring{$T$} if and only if~$f(\gamma_i) \neq 1$ for some~$i \in [k]$. Then, we will transform~$(G',L')$ into an instance~$(G,L)$ of \listHcoloring with the desired properties by replacing edges with the gadgets of Lemma~\ref{lemma:extendedP4exists}, without blowing up the number of vertices.

The \listcoloring{$T$} instance~$(G',L')$ is constructed as follows. Create a path on~$3k$ vertices~$x_1, y_1, z_1, x_2, y_2, z_2, \ldots, x_k, y_k, z_k$. Add vertices~$\gamma_1, \ldots, \gamma_k$ and insert the edge~$\gamma_i y_i$ for all~$i \in [k]$. This defines graph~$G'$. The lists~$L'$ are defined as follows:
\begin{itemize}
	\item $L'(\gamma_i) = \{1,2,3\}$ for~$1 \leq i \leq k$.
	\item $L'(x_i) = \{1,2\}$ for~$2 \leq i \leq k$.
	\item $L'(y_i) = \{1,2,3\}$ for~$1 \leq i \leq k$.
	\item $L'(z_i) = \{3,1\}$ for~$1 \leq i \leq k - 1$.
	\item Finally,~$L'(x_1) = \{2\}$ and~$L'(z_k) = \{3\}$.
\end{itemize}

For this instance~$(G',L')$ of \listcoloring{$T$}, we first argue that a partial coloring that assigns color~$1$ to all of~$\gamma_1, \ldots, \gamma_k$ cannot be extended to a proper list $T$-coloring. To see that, note that due to the edges between~$\gamma_i$ and~$y_i$, the color~$1$ is blocked for all vertices~$y_i$. This means extending the coloring is equivalent to finding a list coloring on the path~$x_1, y_1, z_1, \ldots, x_k, y_k, z_k$, where all $x$-vertices have list~$\{1,2\}$ (except~$x_1$ which must be colored~$2$), where all $y$-vertices have list~$\{2,3\}$, and all $z$-vertices have list $\{3,1\}$ (except~$z_k$ which must be colored~$3$). But the path has no proper list coloring under these conditions: Since the color of~$x_1$ is fixed to~$2$,~$y_1$ must be colored~$3$, implying~$z_1$ must be colored~$1$, which propagates throughout the path to imply that~$y_k$ must be colored~$3$, which conflicts with the fact that~$L'(z_k) = \{3\}$. Hence a mapping that colors all~$\gamma_i$ with~$1$ cannot be extended to a proper list $T$-coloring of~$(G',L')$.

Next, we argue that if~$f \colon \{\gamma_1, \ldots, \gamma_k\} \to \{1,2,3\}$ such that~$f(\gamma_i) \neq 1$ for some~$i \in [k]$, then~$f$ can be extended to a proper list $T$-coloring of~$(G',L')$. Consider such an~$f$, and define~$i^- := \min \{i \mid f(\gamma_i) \neq 1\}$ and~$i^+ := \max \{i \mid f(\gamma_i) \neq 1\}$, which are well-defined. Let~$P$ be the path~$(x_1, y_1, z_1, \ldots, x_k, y_k, z_k)$ in its natural ordering from~$x_1$ to~$z_k$, and extend~$f$ as follows:
\begin{itemize}
	\item Set~$f(y_i) = 1$ for all~$i \in [k]$ for which~$f(\gamma_i) \neq 1$.
	\item For all vertices before~$y_{i^-}$ on~$P$, color the $x$-vertices~$2$, the~$y$-vertices~$3$, and the~$z$-vertices~$1$.
	\item For all vertices after~$y_{i^+}$ on~$P$, color the $x$-vertices~$1$, the~$y$-vertices~$2$, and the~$z$-vertices~$3$.
	\item Consider the vertices we have not assigned a color so far (if any). They form subpaths~$P'$ of~$P$ of the form~$z_j, x_{j+1}, \ldots, x_{j'}$ for~$j < j'$ with~$f(\gamma_j), f(\gamma_{j'}) \neq 1$, while~$f(\gamma_i) = 1$ for~$j < i < j'$. Set~$f(z_j) = 3$, set~$f(x_{j'}) = 2$, and for the remaining vertices of~$P'$ color the $x$-vertices~$2$, the $y$-vertices~$3$, and the~$z$-vertices~$1$. 
\end{itemize}
Note that for all~$\gamma_i$ which are not colored~$1$, the corresponding~$y_i$ gets color~$1$, while if~$f(\gamma_i) = 1$ then~$f(y_i) \in \{2,3\}$. It is straight-forward to verify that the resulting extension of~$f$ forms a proper list $T$-coloring of~$G'$.

To construct the gadget for \listHcoloring promised by the lemma statement, we transform~$(G',L')$ into a \listHcoloring instance~$(G,L)$ as follows.
Let $(a,b,c,d,e)$ be an extended $P_4$ gadget for $H$ as guaranteed by \cref{lemma:extendedP4exists}, and let~$(G_{a,c,e}, L_{a,c,e})$ with distinguished vertices~$\gamma^*_1, \gamma^*_2$.

\begin{itemize}
	\item Initialize~$(G,L)$ as a copy of~$(G',L')$. Replace occurrences of color~$1$ by~$c$, of color~$2$ by~$a$, and of color~$3$ by~$e$.
	\item For each edge~$e$ of~$G'$, do the following. Let~$v_1, v_2$ be the endpoints of~$e$. Remove edge~$v_1 v_2$ from~$G'$, insert a new copy of the graph~$(G_{a,c,e}, L_{a,c,e})$ with lists as given by~$L_{a,c,e}$. Let~$\gamma^*_1, \gamma^*_2$ denote the distinguished vertices of the inserted copy. Identifying~$\gamma^*_1$ with~$v_1$ and~$\gamma^*_2$ with~$v_2$.
\end{itemize}
Since~$G$ has~$\Oh(k)$ vertices and edges, the transformation to~$G'$ introduces~$\Oh(k)$ gadgets, each of which has constant size. Hence~$|V(G')| \in \Oh(k)$, as required. It is easy to perform the construction in polynomial time. Since the gadget~$(G_{a,c,e}, L_{a,c,e})$ for distinguished vertices~$\gamma^*_1, \gamma^*_2$ for \listHcoloring has the same effect as an edge in \listcoloring{$T$}, while color~$1 \in V(T)$ was mapped to color~$c \in V(H)$, it follows that a mapping~$f \colon \{\gamma_1, \ldots, \gamma_k\} \to \{a,c,e\}$ can be extended to a proper list $H$-coloring of~$(G,L)$ if and only if~$f(\gamma_i) \neq c$ for some~$c \in [k]$. Since~$G$ is built by replacing all edges of~$G'$ by gadgets, which are consistent instances by Lemma~\ref{lemma:extendedP4exists}, and since graph~$G'$ we start from is a tree and therefore bipartite, it is easy to see that the instance~$(G,L)$ is consistent. This concludes the proof.
\end{proof}

\newcommand{\acegadget}{$\cola,\colc,\cole$-NOT-gadget\xspace}
\newcommand{\bdgadget}{$\colb,\cold$-NOT-gadget\xspace}
%

Using these gadgets in the two-step process described in the beginning of \cref{sec:gadgets}, we now obtain the following. 

\theoremlowerbound*
\begin{proof}
We start by showing that for any connected bipartite graph $H$ that is not a bi-arc graph, \listHcoloring allows no nontrivial sparsification. We  use a linear-parameter transformation from \Anlistcoloring, such that the lower bound follows from \cref{lem:annotated-coloring:LB}. 

Since $H$ is bipartite and not a bi-arc graph, 
it is not the complement of a circular arc graph~\cite{FederHH03}, and it follows from \cref{thm:NP-hard-circular-arc} that $H$ has an induced cycle of length at least six or an asteroid. It then follows from \cref{lemma:extendedP4exists} that $H$ has an extended $P_4$ gadget on distinguished vertices $(\cola,\colb,\colc,\cold,\cole)$ of $H$.
Furthermore, there exist two relevant gadgets as described by \cref{lemma:extendedP4exists}. We call the gadget constructed for $Q=\{\cola,\colc,\cole\}$ the 
 \acegadget, and the one constructed for $Q=\{\colb,\cold\}$ the \bdgadget.

Let an instance $(G,\ELL,\mathcal{S},F)$ of \Anlistcoloring be given, we show how to create an instance $\tilde{G}$ of \listHcoloring. Initialize $\tilde{G}$ as $G$ (ignoring the annotations), where every vertex in $\tilde{G}$ receives the same list it had in $G$, where now $\cola,\colb,\colc,\cold,\cole$ refer to the vertices of the extended $P_4$ gadget present in $H$. For any $\{u,v\} \in F$, if $\ELL(u)\subseteq \{\cola,\colc,\cole\}$ (implying also $\ELL(v) \subseteq \{\cola,\colc,\cole\}$), add a new \acegadget to $\tilde{G}$. Otherwise, meaning that $\ELL(u)\subseteq \{\colb,\cold\}$ and $\ELL(v) \subseteq \{\colb,\cold\}$, we add a new \bdgadget to $\tilde{G}$.
Identify vertex $\gamma_1$ of the added gadget with $u$, and vertex $\gamma_2$ with $v$.

For every $S = \{s_1,\ldots,s_m\}\in \mathcal{S}$, add a new gadget as described by \cref{lemma:blocking:gadget} for $k=m$ to $\tilde{G}$. Note that such a gadget has  $\Oh(m)$ vertices. Identify vertex $\gamma_i$ of the gadget with vertex $s_i$ for all $i\in[m]$.

It is easy to observe from the correctness of the added gadgets, that $\tilde{G}$ is list $H$-colorable if and only if $G$ had a coloring respecting the annotations.

We continue by bounding the number of vertices in $\tilde{G}$. Using that $\sum_{S \in \mathcal{S}} |S| \leq 3|V(G)|$ and $|F| \leq |V(G)|$ by definition of \Anlistcoloring, we get
\[|V(\tilde{G})| = \underbrace{|V(G)|}_{\text{init}} + \underbrace{|V(G)| \cdot \Oh(1) }_{\text{NOT-gadgets}} + \underbrace{\Oh(|V(G)|)}_{\text{\cref{lemma:blocking:gadget}}} = \Oh(|V(G)|), \]
which is properly bounded for a linear parameter transformation. The result for bipartite graphs $H$ thus follows from \cite[Theorem 3.8]{BodlaenderJK14}. Observe that the constructed graph $\tilde{G}$ is consistent, such that the lower bound holds even for consistent instances of \listHcoloring.

It remains to show the result for non-bipartite graphs $H$ and bipartite graphs $H$ that are not connected. We start with the latter. Since $H$ is not a bi-arc graph, there must exist a connected component $H'$ of $H$ such that $H'$ is not a bi-arc graph (and since $H$ is bipartite, $H'$ is bipartite). This follows from the fact that by \cref{thm:NP-hard-circular-arc} the graph $H$ has an induced cycle of length at least six or an asteroid, and this structure must be found in one of its components. It now follows from the above, that \textsc{List $H'$-Coloring} does not have a generalized kernel of size $\Oh(n^{2-\varepsilon})$, unless containment. There is a straightforward linear-parameter transformation from \textsc{List $H'$-Coloring} to \textsc{List $H$-Coloring}, taking the exact same instance and using the lists to ensure that only colors from $H'$ can be used to color each vertex. Therefore, the lower bound for \textsc{List $H$-Coloring} for (possibly not connected) bipartite graphs $H$ follows.

We conclude the proof by showing the result for non-bipartite graphs. Let $H$ be an undirected graph that is not a bi-arc graph, such that $H$ is non-bipartite, let $H^*$ be the associated bipartite graph of $H$.   Since $H$ is not a bi-arc graph, it follows that $H^*$ is not the complement of a circular arc graph~\cite[Proposition 3.1]{FederHH03}. Since $H^*$ is bipartite and irreflexive it follows that $H^*$ is not a bi-arc graph.

As proven above, it follows that \textsc{List $H^*$-Coloring} does not have a generalized kernel of size $\Oh(n^{2-\varepsilon})$, unless \containment. \cref{prop:bipartite-associted} gives a straightforward linear-parameter transformation from \textsc{List $H^*$-Coloring} to \listHcoloring, showing that the same lower bound holds for \listHcoloring.
\end{proof}

\section{Conclusion} \label{sec:conclusion}

%

A natural open question is whether analogous results can be obtained for the (non-list) \Hcoloring problem.
Despite the obvious similarity of \Hcoloring and \listHcoloring, they appear to behave very differently when it comes to proving lower bounds. 
All hardness proofs for \listHcoloring~\cite{FederH98,Feder99,FederHH03,FEDER2007386,Okrasa20}, including the proofs in this paper, are purely combinatorial and focus on the \emph{local} structure of $H$.
In all of them, we first identify some ``hard'' substructure $H'$ in $H$, and then prove the lower bound for $H'$.
This can be done, as we can ignore vertices in $V(H) \setminus V(H')$ by not including them in the lists.
On the other hand, all proofs for \Hcoloring use some algebraic tools~\cite{Bulatov05,HellN90,DBLP:conf/soda/OkrasaR20,Siggers09} which allow capturing the \emph{global} structure of $H$. We therefore expect similar difficulties in the case of proving sparsification lower bounds for \Hcoloring.

\bibliography{referencescrosscomposition}
\end{document}